\documentclass[12pt,a4paper]{article}

\usepackage{amsmath,amsthm,amssymb,amscd, a4wide}

\usepackage{dsfont}
\usepackage{mathrsfs}%
\usepackage{setspace}
\usepackage{hyperref}
\usepackage{enumerate} %!!!!!!!!!!!!!!!!!!!%%
\usepackage{mathtools} %for shortintertext%%%
\usepackage{latexsym,amsfonts,bbm, mathabx} %für schönere Schrift
\newcommand{\ljomega}{\hat l^{j,\omega}}
\newcommand{\lNjomega}{l^{j,\omega}_N}

\newcommand{\ud}{\mathrm{d}}

\DeclareMathOperator{\supp}{supp}

\numberwithin{equation}{section}

\newtheorem{theorem}{Theorem}[section]
\newtheorem{lemma}[theorem]{Lemma}
\newtheorem{prop}[theorem]{Proposition}
\newtheorem{cor}[theorem]{Corollary}
\newtheorem{remark}[theorem]{Remark}

\theoremstyle{definition}

\newtheorem{assumptions}[theorem]{Assumptions}

\newcommand{\laplace}{- \mathop{}\!\mathbin\bigtriangleup}

\newcommand{\E}{\mathds{E}}

\newcommand{\llargest}{l_{N,>}^{1,\omega}}
\newcommand{\llargestzwei}{l_{N,>}^{2,\omega}}

\newcommand{\e}{\mathrm{e}}
\newcommand{\llargesttilde}{\tilde l_{N,>}^{\,1,\omega}}
\newcommand{\llargesttildezwei}{\tilde l_{N,>}^{\,2,\omega}}
\newcommand{\llargestk}[1]{l_{k,>}^{#1,\omega}}
\newcommand{\llargestkminuszwei}[1]{l_{k-2,>}^{#1,\omega}}

\numberwithin{equation}{section}

\begin{document}
	
	\allowdisplaybreaks[1]
	
	\thispagestyle{empty}
	
	\vspace*{1cm}
	
	\begin{center}
		
		{\Large \bf  On a condition for type-I Bose--Einstein condensation in random potentials in $d$ dimensions} \\

		\vspace*{1cm}
		
		{\large  Joachim~Kerner \footnote{E-mail address: {\tt joachim.kerner@fernuni-hagen.de}} }%
		
		\vspace*{5mm}
		
		Department of Mathematics and Computer Science\\
		FernUniversität in Hagen\\
		58084 Hagen\\
		Germany\\
		
		\vspace*{1cm}
		
		{\large Maximilian~Pechmann \footnote{E-mail address: {\tt mpechmann@utk.edu}}}
		
		\vspace*{5mm}
		
		Department of Mathematics\\
		University of Tennessee\\
		Knoxville, TN 37996\\
		USA
		
			\vspace*{1cm}
			
			{\large  Wolfgang~Spitzer \footnote{E-mail address: {\tt wolfgang.spitzer@fernuni-hagen.de}}}%
		
			\vspace*{5mm}
			
			Department of Mathematics and Computer Science\\
			FernUniversität in Hagen\\
			58084 Hagen\\
			Germany\\

	\end{center}
	
	\vfill
	
	\begin{abstract} In this paper we discuss Bose--Einstein condensation (BEC) in systems of pairwise non-interacting bosons in random potentials in $d$ dimensions. Working in a rather general framework, we provide a ``gap condition'' which is sufficient to conclude existence of type-I BEC in probability and in the $r$th mean. We illustrate our results in the context of the well-known (one-dimensional) Luttinger--Sy model. Here, whenever the particle density exceeds a critical value, we show in addition that only the ground state is macroscopically occupied.
	\end{abstract}
	
	\newpage

	\section{Introduction}
	
	Bose--Einstein condensation (BEC) is usually referred to as a macroscopic occupation of a one-particle state. Instead, a broader definition of BEC speaks of generalized Bose–Einstein condensation (g-BEC) where only a macroscopic occupation of an arbitrarily small energy band of one-particle states is required. Depending on the quantity of macroscopically occupied one-particle states in the condensate one then distinguishes three types of g-BEC: A condensation of type I is said to occur if the number of macroscopically occupied one-particle states is finite but at least one. If there are infinitely many such one-particle states, the condensation is said to be of type II. Lastly, a generalized condensate in which none of the one-particle states are macroscopically occupied is called a type-III condensate.
	
	Although it is generally believed that repulsive interactions between the particles should be taken into account when considering Bose gases in random potentials \cite{lenoble2004bose} (see also \cite{stolz1995localization}, \cite{germinet2005localization}, \cite{klein2007localization}, \cite{SeiringerWarzel}), the study of non-interacting Bose gases (Bose gases without interaction between the particles) is nevertheless important since general features of BEC in the considered environment are revealed \cite{lenoble2004bose}. Generally, a main step in the proof of the occurrence of g-BEC in non-interacting Bose gases is to show that a certain critical density is finite. To determine the type of the condensate, however, is more difficult since this task requires a fairly accurate knowledge on how the eigenvalues of the corresponding one-particle Hamiltonian at the bottom of the spectrum behave in the thermodynamic limit.
	
	In 1973 and 1974, Kac and Luttinger explored a non-interacting Bose gas in three dimensions in the presence of randomly distributed impurity centers which were assumed to represent hardcore potentials of finite range \cite{kac1973bose,kac1974bose}. As a result, Kac and Luttinger were able to show that g-BEC occurs in probability for sufficiently low temperatures or, equivalently, for sufficiently large particle densities. Furthermore, they \textit{conjectured} the existence of a macroscopic occupation of the ground state in this case but could not prove their assertion. They claimed that showing a certain lower bound for the gap between the two lowest eigenvalues of the respective one-particle random Schrödinger operator of the Bose gas would be sufficient to conclude a macroscopic occupation of the ground state. However, no further explanation was given to bolster this claim. Although it is indeed fairly easy to exclude a macroscopic occupation of excited states with the energy gap as assumed by Kac and Luttinger, it seems difficult to infer that a type-I BEC necessarily occurs. We suspect that Kac and Luttinger were not aware of the possibility that a type-III BEC \cite{van1982generalized, van1983condensation, van1986general, van1986general2} can occur.
	
	In this paper, we rigorously prove that a slightly stronger lower bound for the energy gap compared to the one assumed by Kac and Luttinger implies a macroscopic occupation of the ground state in probability and in the $r$th mean. Consequently, confirming this lower bound for a given random model then proves the so-called \textit{Kac--Luttinger conjecture} in that case. Our results are valid in all spatial dimensions $d \in \mathds{N}$. We formulate our requirements for the random potential in a rather general framework. All of them are fulfilled, for example, by a Poisson random potential on $\mathds R^d$, $d \in \mathds{N}$, that has a non-negative, bounded, and compactly supported function as its single-impurity potential. The one-dimensional Luttinger--Sy \cite{luttinger1973bose,luttinger1973low} model (with infinite as well as finite interaction strength, see also \cite{KPS182}) also meets the requirement. 
	
	We formulate our general model and our assumptions for the random potential in Section~\ref{Section Main Results}. There, we also state and prove our main results, i.e., Theorem~\ref{MainResult} and Corollary~\ref{Corollary type-I BEC in rth mean}. In Section~\ref{section Luttinger Sy infinite strength}, we apply our results of Section~\ref{Section Main Results} to the Luttinger--Sy model and show that in this model a type-I BEC, in which only the ground state is macroscopically occupied, occurs in probability and in the $r$th mean. For the convenience of the reader we provide miscellaneous results that we use for proofs in this paper in Appendix~\ref{Miscellaneous results}.

Finally, we remark that type-I BEC in the Luttinger--Sy model was previously investigated in \cite{LenobleZagrebnovLuttingerSy}. However, the authors used other methods and worked in a somewhat different setting resulting in a different thermodynamic limit. 

\section{General model and results} \label{Section Main Results}

Let $(\Omega,\mathscr A,\mathds{P})$ be a probability space. In this paper we shall consider self-adjoint random Schrödinger operators (acting as single-particle Hamiltonians) of the form 
\begin{equation}
H_{\omega}=-\Delta+V_{\omega}\ ,\quad \omega \in \Omega\ ,
\end{equation}
on $L^2(\mathds{R}^d)$. The restriction of $H_{\omega}$ to a box $\Lambda_N:=\big( -L^{1/d}_N/2,L^{1/d}_N/2 \big)^{d}$, $L_N > 0$, with Dirichlet boundary conditions is given by the self-adjoint operator 
\begin{equation}\label{FVH}
H_{N,\omega}=- \Delta+V_{\omega} \ ,\quad \omega \in \Omega\ ,\quad N \in \mathds{N}\ ,
\end{equation}
on $L^2(\Lambda_N)$ with the Sobolev space $H_0^1(\Lambda_N)$ as its domain. Here, $-\Delta$ is the $d$-dimensional (Dirichlet-) Laplacian and $V_{\omega}$ is a random potential \cite[Chapter 1]{pastur1992spectra}.
We assume that $H_{N,\omega}$ is $\mathds P$-almost surely a positive self-adjoint operator with purely discrete spectrum, i.e., with a spectrum that consists only of isolated eigenvalues of finite multiplicities. This is the case, for example, if $V_{\omega}$ is a Poisson random potential on $\mathds R^d$, $d \in \mathds{N}$, with a single-impurity potential that is a bounded, compactly supported and non-negative function \cite[Proposition 3.3]{lenoble2004bose}, \cite[Theorem 5.1]{pastur1992spectra}. These assumptions are also fulfilled if $V_{\omega}$ is, informally, a Poisson random potential on $\mathds R$ with a singular single-impurity potential of the form $\gamma \delta$ where $\gamma > 0$ or $\gamma = \infty$ denotes the interaction strength and $\delta$ the Dirac-$\delta$ distribution \cite[p. 146]{pastur1992spectra}, \cite{LenobleZagrebnovLuttingerSy}. 

We denote the associated sequence of eigenvalues of $H_{N,\omega}$, written in ascending order and each eigenvalue repeated according to its multiplicity, by $(E^{j,\omega}_N)_{j \in \mathds{N}}$ and the associated sequence of eigenfunctions by $(\varphi^{j,\omega}_N  )_{j \in \mathds{N}}$. Note that $E^{j,\omega}_N > 0$ and $\varphi^{j,\omega}_N \in L^2(\Lambda_N)$ for $\mathds P$-almost all $\omega \in \Omega$ and all $j,N \in \mathds{N}$. Moreover, we indicate a random variable by writing an $\omega$ as a subscript or superscript. $\mathds{E}$ refers to the expectation of a random variable with respect to $\mathds{P}$.
\begin{remark} We also allow for Hamiltonians $H_{N,\omega}$ where the potential $V_{N,\omega}$ is distribution-valued as it is the case in the Luttinger--Sy model. In this case, one defines the Hamiltonian $H_{N,\omega}$ rigorously via its associated quadratic form on $H_0^1(\Lambda_N)$. 
\end{remark}
We recall that the standard counting function is given by
\begin{equation}
\mathcal{N}^{\mathrm I,\omega}_{N}(E):=\frac{1}{L_N} \left| \left\{ j \in \mathds{N}: E^{j,\omega}_N < E \right\} \right|\ , \quad  \omega \in \Omega \ , \quad E \in \mathds R \ ,
\end{equation}
where $\left|A\right|$ denotes the number of elements of a set $A \subset \mathds{N}$. The function $\mathcal N_N^{\mathrm{I,\omega}}$ is called the \emph{integrated density of states} and is $\mathds P$-almost surely left-continuous and monotonically increasing. For $\mathds P$-almost all $\omega \in \Omega$, $\mathcal N_N^{\mathrm{I,\omega}}$ consequently defines a unique Borel measure, that is, a measure on the Borel-$\sigma$-algebra on $\mathds R$ such that the measure of any compact subset of $\mathds R$ is finite. Ones denotes this uniquely defined measure by $\mathcal N_N^{\omega}$ and calls it the \emph{density of states}.

Throughout this paper, we employ the standard \emph{thermodynamic limit}: For $N \to \infty$ the volume $L_N$ of the box $\Lambda_N$ increases with the particle number $N$ of the system in such a way that the \emph{particle density} $N / L_N =: \rho > 0$ remains constant for all $N \in \mathds{N}$. \\

\noindent In addition to what was mentioned above, we assume the following about $H_{N,\omega}$ in the rest of the paper:
\begin{assumptions} \label{assumptions} \ 
	
	\begin{enumerate}[(i)]
	    \item \label{assumption 1} For the ground-state energy $E_N^{1,\omega}$ of $H_N^{1,\omega}$, we have $\lim\limits_{N \to \infty} E_N^{1,\omega} = 0$ $\mathds P$-almost surely. 
	    \item \label{assumption 2} There exists a non-random measure $\mathcal N_{\infty}$ such that $\mathds P$-almost surely
		\begin{equation}
		\lim\limits_{N \to \infty} \mathcal N_N^{\omega} = \mathcal N_{\infty}
		\end{equation}
		in the vague sense; see Remark~\ref{DefVagueConvergence}.
		
		\item \label{assumption 3} There exist constants $c_1 > 0$ and $\widetilde{E} > 0$ such that for all $N \in \mathds{N}$ and all $0 < E \leq \widetilde{E}$ one has
		\begin{equation*}
		\mathds{E}\ \mathcal{N}^{\mathrm I,\omega}_{N}(E)\leq c_1\mathcal{N}^{\mathrm I}_{\infty}(E) \ .
		\end{equation*}
		Here, 
		$$\mathcal{N}^{\mathrm I}_{\infty}(E):= \mathcal N_{\infty}((-\infty,E)) =
		\begin{cases}
		\int\limits_{(0,E)} \ud \mathcal{N}_{\infty}(E) \quad & \text{ if } E > 0 \\
		0 \quad & \text{ if } E \le 0
		\end{cases}$$
		for all $E \in \mathds R$.
		\item \label{assumption 4} There exists a constant $0 < \eta_1 < 1$ such that 
		\begin{equation*}
		\lim_{N \rightarrow \infty} N^{1-\eta_1}\mathcal{N}^{\mathrm I}_{\infty}\left(\left(\frac{(1+\eta_1/2)\nu \gamma^{d/2}_{d}}{\ln L_N}\right)^{2/d} \right)=0\ ,
		\end{equation*} 
		where $\gamma_d > 0$ is the lowest Dirichlet eigenvalue of $-\Delta$ on the $d$-dimensional sphere with unit volume.
	\end{enumerate}
\end{assumptions}

We refer to Remark \ref{remark Poisson erfuellt bedinungen} for examples of random potentials for which these assumptions are met.

\begin{remark}\label{DefVagueConvergence} A sequence $(\mathscr M_N)_{N \in \mathds{N}}$ of Radon measures (that is, of measures that are inner regular and locally finite) on the Borel-$\sigma$-algebra on $\mathds R$ is said to converge to a Radon measure $\mathscr M_{\infty}$ in the vague sense if $$\lim_{N \to \infty} \int\limits_{\mathds R} f(E) \, \ud \mathscr M_N(E) = \int\limits_{\mathds R} f(E) \, \ud \mathscr M_{\infty}(E)$$ for all continuous and compactly supported functions $f : \mathds R \to \mathds R$, see, e.g., \cite{bauer2001measure} for more details.
\end{remark}

\begin{remark} \label{assumptions remark}
 Assumption \ref{assumptions}~\eqref{assumption 4} implies in particular that
		\begin{equation*}
		\lim_{\epsilon \rightarrow 0^+}\epsilon^{-1}\mathcal{N}^{\mathrm I}_{\infty}(\epsilon)=0\ ,
		\end{equation*}
		as well as
		\begin{equation*}
		\int\limits_{0}^{\varepsilon} \mathcal{N}^{\mathrm I}_{\infty}(E)E^{-2}\, \ud E < \infty\ 
		\end{equation*}
		for arbitrary $\varepsilon > 0$.
\end{remark}

An important quantity in the context of BEC in non-interacting models is the so-called \textit{critical density} $\rho_{c}(\beta)$ where $\beta:=1/T \in (0,\infty)$ denotes the inverse temperature. With
\begin{equation} \label{definition B(E)}
\mathcal{B}(E):= \left( \mathrm{e}^{\beta E}-1 \right)^{-1} \, \mathds 1_{(0,\infty)}(E) \ ,
\end{equation}
$\mathds 1_{A}(E)$ being the indicator function of a measurable set $A \subset \mathds R$,
$\rho_{c}(\beta)$ is defined as
\begin{equation} \label{definition critical density}
\rho_{c}(\beta):=\int\limits_{\mathds{R}} \mathcal{B}(E)\, \ud \mathcal{N}_{\infty}(E) \ .
\end{equation}
As usual, we study BEC in the grand-canonical ensemble of statistical mechanics. Hence, the number of particles occupying the $j$th eigenstate is $\mathds P$-almost surely given by
\begin{equation} \label{definition occupation numbers grand canonical ensemble}
n^{j,\omega}_N:= \left( \mathrm{e}^{\beta(E_N^{j,\omega}-\mu^{\omega}_N)}-1 \right)^{-1}\ .
\end{equation}
Here, $\mu^{\omega}_N \in (-\infty,E^{1,\omega}_N)$ is the so-called chemical potential which is, for $\mathds P$-almost all $\omega \in \Omega$, uniquely determined by the condition that
\begin{equation}
\frac{1}{L_N}\sum_{j=1}^{\infty}n^{j,\omega}_N= \rho
\end{equation}
holds for all values of $N \in \mathds{N}$. 
We now state a first result which shows that the critical density as defined in \eqref{definition critical density} is indeed finite which consequently implies the existence of generalized BEC. For the proof we refer to \cite[Theorem 4.1]{lenoble2004bose} in combination with \cite[Theorem 2.7]{KPS182}.
\begin{theorem}[Generalized BEC]\label{TheoremGENBEC} Under the above Assumptions~\ref{assumptions} one has $\rho_{c}(\beta) < \infty$. In addition, if and only if the particle density is larger than the critical density, $\rho > \rho_{c}(\beta)$, generalized BEC $\mathds P$-almost surely occurs, i.e., one has
	\begin{equation}
	\mathds{P}\left(\lim_{\epsilon \rightarrow 0^+} \liminf_{N \rightarrow \infty}\frac{1}{N} \sum_{j\in \mathds{N}: E^{j,\omega}_N-E^{1,\omega}_N \leq \epsilon}n^{j,\omega}_N=\frac{\rho-\rho_c(\beta)}{\rho}>0 \right)=1\ .
	\end{equation}
Moreover, whenever $\rho > \rho_{c}(\beta)$ the sequence of chemical potentials $(\mu^{\omega}_N)_{N \in \mathds{N}}$ $\mathds P$-almost surely converges to zero.
\end{theorem}
In order to study the existence of macroscopic occupation of the ground state and to determine the type of condensation, we introduce the event
\begin{equation}\label{Definition Omega4}
\Omega_N^{c_2,c_3}:=\left\{\omega \in \Omega: E^{c_2+1,\omega}_N-E^{1,\omega}_N \geq c_3 N^{-1+\eta_1}\ \text{and} \ E^{1,\omega}_N \leq \left[\left(1+\frac{\eta_1}{4}\right)\frac{\nu \gamma^{d/2}_d}{\ln L_N} \right]^{2/d} \right\}
\end{equation}
for some $c_2 \in \mathds{N}$ and $c_3 > 0$ and with the constant $\eta_1$ from Assumption~\ref{assumptions}~\eqref{assumption 4}. At this point we remark two facts: Firstly, the event $\Omega_N^{c_2, c_3}$ is the main ingredient in the ``gap condition'' as formulated in Theorem~\ref{MainResult}. Secondly, the upper bound for the ground-state energy $E_N^{1,\omega}$ in \eqref{Definition Omega4} is $\mathds P$-almost surely fulfilled if $V_{\omega}$ is a Poisson random potential on $\mathds R^d$ with a single-impurity potential that is a bounded, non-negative, compactly supported function, see for example \cite[Theorem 4.6]{sznitman1998brownian}.
\begin{prop}\label{Proposition 1} Under Assumptions~\ref{assumptions}, for $\rho > \rho_{c}(\beta)$, one has
	\begin{equation}
	\liminf_{N \rightarrow \infty} \mathds{E}\, \int\limits_{(0,E^{c_2,\omega}_N]}\mathcal{B}(E-\mu^{\omega}_N)\, \ud \mathcal{N}^{\omega}_N(E) \geq \rho - \rho_{c}(\beta)-\rho\left(1-\liminf_{N\rightarrow \infty} \mathds{P}(\Omega_N^{c_2, c_3})  \right)
	\end{equation}
\end{prop}
\begin{proof}
For $\epsilon > 0$ and $\eta_1$ from \eqref{Definition Omega4} we define the sets
	\begin{align*}
	\widehat{\Omega}_N^{(1),\epsilon} & := \left\{ \omega \in \Omega : E_N^{c_2,\omega} \ge \epsilon \right\} \ ,\\
	\widehat{\Omega}_N^{(2),\epsilon} & := \left\{\omega \in \Omega : E_N^{c_2,\omega} < \epsilon \right\} \ , \\
	\widehat{\Omega}_N^{c_2, c_3} & := \widehat{\Omega}_N^{(2),\epsilon} \cap \Omega_N^{c_2, c_3} \cap \left\{ \omega \in \Omega : E_N^{c_2 + 1,\omega} \ge \left[ \left( 1 + \dfrac{\eta_1}{2} \right) \dfrac{\nu \gamma_d^{d/2}}{\ln(L_N)} \right]^{2/d} \right\} \ , \\
	\shortintertext{and}
	\widetilde{\Omega}_N^{c_2, c_3} & := \widehat{\Omega}_N^{(2),\epsilon} \cap \Omega_N^{c_2, c_3} \cap  \left\{ \omega \in \Omega : E_N^{c_2 + 1,\omega} < \left[ \left( 1 + \dfrac{\eta_1}{2} \right) \dfrac{\nu \gamma_d^{d/2}}{\ln(L_N)} \right]^{2/d} \right\}\ .
	\end{align*}
	For all $\epsilon > 0$ and all $N \in \mathds{N}$ we conclude
	\begin{align*}
	\rho & = \int\limits_{(0,\infty)} \mathcal{B}(E - \mu_N^{\omega}) \, \ud \mathcal N_N^{\omega}(E)\\
	& = \int\limits_{(0,\epsilon]} \mathcal{B}(E - \mu_N^{\omega}) \, \ud \mathcal N_N^{\omega}(E) + \int\limits_{(\epsilon,\infty)} \mathcal{B}(E - \mu_N^{\omega}) \, \ud \mathcal N_N^{\omega}(E)
	\end{align*}
	and thus
	\begin{align*}
	\rho & \le \int\limits_{\widehat \Omega_N^{(1),\epsilon}}  \int\limits_{\left( 0,E_N^{c_2,\omega} \right]} \mathcal{B}(E - \mu_N^{\omega}) \, \ud \mathcal N_N^{\omega}(E)\, \ud \mathds P(\omega) \\
	& \qquad \qquad + \, \int\limits_{\widehat \Omega_N^{(2),\epsilon}} \int\limits_{\left( 0,E_N^{c_2,\omega} \right]} \mathcal{B}(E - \mu_N^{\omega}) \, \ud \mathcal N_N^{\omega}(E)\, \ud \mathds P(\omega) \\
	& \qquad \qquad + \, \int\limits_{\widehat \Omega_N^{(2),\epsilon}}  \int\limits_{(E_N^{c_2,\omega},\epsilon]} \mathcal{B}(E - \mu_N^{\omega}) \, \ud \mathcal N_N^{\omega}(E) \, \ud \mathds P(\omega) \\
	& \qquad \qquad + \E  \int\limits_{(\epsilon,\infty)} \mathcal{B}(E - \mu_N^{\omega}) \, \ud \mathcal N_N^{\omega}(E)\ .
	\end{align*}
	Hence, we obtain
	\begin{align}
	& \E \int\limits_{\left( 0,E_N^{c_2,\omega} \right]} \mathcal{B}(E - \mu_N^{\omega}) \, \ud \mathcal N_N^{\omega}(E)\\
	& \quad \ge \rho - \int\limits_{\widehat{\Omega}_N^{(2),\epsilon} \cap \, \Omega_N^{c_2, c_3}}  \int\limits_{(E_N^{c_2,\omega},\epsilon]} \mathcal{B}(E - \mu_N^{\omega})\, \ud \mathcal N_N^{\omega}(E)\, \ud \mathds P( \omega) \label{proof BEC type ns LSM 1}\\
	%& \qquad \qquad - \int\limits_{\Omega_N^{c_2, c_3}} \left[ \, \int\limits_{\left( \left[ (1+ \eta_1/2) \nu \gamma_d^{d/2} / \ln(L_N) \right]^{2/d},\epsilon \right]} \mathcal{B}(E - \mu_N^{\omega}) \, \ud \mathcal N_N^{\omega}(E) \right] \, \ud \mathds P( \omega) \label{proof BEC type ns LSM 2} \\
	& \qquad \qquad - \int\limits_{\widehat{\Omega}_N^{(2),\epsilon} \cap \, \Omega \backslash \Omega_N^{c_2, c_3}} \int\limits_{(E_N^{c_2,\omega},\epsilon]} \mathcal{B}(E - \mu_N^{\omega}) \, \ud \mathcal N_N^{\omega}(E)\, \ud \mathds P( \omega) \label{proof BEC type ns LSM 3} \\
	& \qquad \qquad  - \E \int\limits_{(\epsilon,\infty)} \mathcal{B}(E - \mu_N^{\omega}) \, \ud \mathcal N_N^{\omega}(E)\label{proof BEC type ns LSM 4} \\
	& \quad =: \rho - A_N(\epsilon) - B_N(\epsilon) - C_N(\epsilon)
	\end{align}
	for all $\epsilon > 0$ and all $N \in \mathds{N}$. In the following, we will discuss the limit $\lim_{\epsilon \to 0+} \limsup_{N \to \infty}$ of this equation.
	
	Firstly, we state an upper bound for the term in line \eqref{proof BEC type ns LSM 4} in the considered limit. Let $\epsilon > 0$ be arbitrarily given. For $\mathds P$-almost all $\omega \in \Omega$ and all $N \in \mathds{N}$, we have the bound $\int_{(\epsilon,\infty)} \mathcal{B}(E - \mu_N^{\omega}) \, \mathrm{d} \mathcal N_N^{\omega} (E) \leq \rho$. %(Wieso wird das gebraucht? Für das (reverse) Fatou Lemma!).
	Employing Lemma~\ref{Lemma 2 Gen BEC}, we conclude that
	\begin{align*}
	\limsup\limits_{N \to \infty} \int\limits_{(\epsilon,\infty)} \mathcal{B}(E - \mu_N^{\omega}) \, \mathrm{d} \mathcal N_N^{\omega} (E) \le \int\limits_{(\epsilon,\infty)} \mathcal{B}(E) \, \mathrm{d} \mathcal N_{\infty} ( E) + \dfrac{2}{\beta \epsilon} \mathcal N_{\infty}^{\mathrm{I}}(\epsilon)
	\end{align*}
	$\mathds P$-almost surely. By using the (reverse) Fatou Lemma, we obtain
	\begin{align*}
	\limsup\limits_{N \to \infty} \mathds E \int\limits_{(\epsilon,\infty)} \mathcal{B}(E - \mu_N^{\omega}) \, \mathrm{d} \mathcal N_N^{\omega} (E)\le \int\limits_{(\epsilon,\infty)} \mathcal{B}(E) \, \mathrm{d} \mathcal N_{\infty} ( E) + \dfrac{2}{\beta \epsilon} \mathcal N_{\infty}^{\mathrm{I}}(\epsilon) \ .
	\end{align*}
	Hence,
	\begin{align*}
	\lim\limits_{\epsilon \to 0+} \limsup\limits_{N \to \infty} C_N(\epsilon) %& = \lim\limits_{\epsilon \to 0+} \limsup\limits_{N \to \infty} \mathds E  \int\limits_{(\epsilon,\infty)} \mathcal{B}(E - \mu_N^{\omega}) \, \mathrm{d} \mathcal N_N^{\omega} (E) 	\\
	& \le \lim\limits_{\epsilon \to 0+} \int\limits_{(\epsilon,\infty)} \mathcal{B}(E) \, \mathrm{d} \mathcal N_{\infty} ( E) = \int\limits_{(0,\infty)} \mathcal{B}(E) \, \mathrm{d} \mathcal N_{\infty} ( E) = \rho_c(\beta) \ ,
	\end{align*}
	where we used definitions \eqref{definition B(E)} and \eqref{definition critical density} in the last step. 
	
	Next, we are concerned with the integral in line \eqref{proof BEC type ns LSM 1}, and we are going to show that $\lim_{\epsilon \to 0+} \lim_{N \to \infty} A_N(\epsilon) = 0$. Let $\epsilon > 0$ be arbitrarily given. For convenience, we define
	\begin{align*}
	A_N^{(1)}(\epsilon) & := \int\limits_{\widetilde \Omega_N^{c_2, c_3}}  \int\limits_{\left( E_N^{c_2,\omega},\left[ (1+ \eta_1/2) \nu \gamma_d^{d/2} / \ln(L_N) \right]^{2/d} \right]} \mathcal{B}(E - \mu_N^{\omega}) \, \mathrm{d} \mathcal N_N^{\omega} (E) \, \ud \mathds P( \omega) \ ,\\
	A_N^{(2)}(\epsilon) & :=  \int\limits_{\widetilde \Omega_N^{c_2, c_3}} \int\limits_{\left( \left[ (1+ \eta_1/2) \nu \gamma_d^{d/2} / \ln(L_N) \right]^{2/d},\epsilon \right]}\mathcal{B}(E - \mu_N^{\omega}) \, \mathrm{d} \mathcal N_N^{\omega} (E) \, \ud \mathds P( \omega)
	\shortintertext{for all $N \in \mathds N$ large enough such that $[ (1+ \eta_1/2) \nu \gamma_d^{d/2} / \ln(L_N) ]^{2/d} < \epsilon$, and}
	A_N^{(3)}(\epsilon) & := \int\limits_{\widehat \Omega_N^{c_2, c_3}}  \int\limits_{(E_N^{c_2,\omega},\epsilon]} \mathcal B(E - \mu_N^{\omega}) \, \ud \mathcal N_N^{\omega}(E) \mathrm{d} \mathds P( \omega) \ .
	\end{align*}
Now, since $\mathcal{B}(E - \mu_N^{\omega}) \le [ \beta ( E - E_N^{1,\omega} ) ]^{-1} \le (\beta c_3)^{-1} N^{1-\eta_1}$
	for all but finitely many $N \in \mathds{N}$, all $\omega \in \Omega_N^{c_2, c_3}$, and all $E \ge E_N^{c_2 + 1,\omega}$, we obtain
	\begin{align*}
	\lim\limits_{N \to \infty} A_N^{(1)}(\epsilon) & = \lim\limits_{N \to \infty} \int\limits_{\widetilde \Omega_N^{c_2, c_3}} \int\limits_{\left( E_N^{c_2,\omega},\left[ (1+ \eta_1/2) \nu \gamma_d^{d/2} / \ln(L_N) \right]^{2/d} \right]} \mathcal{B}(E - \mu_N^{\omega}) \, \mathrm{d} \mathcal N_N^{\omega} (E) \, \ud \mathds P( \omega) \\
	%& \le \lim\limits_{N \to \infty} \int\limits_{\widetilde \Omega_N^{c_2, c_3}} \left[ \, \int\limits_{\left( E_N^{c_2,\omega},\left[ (1+ \eta_1/2) \nu \gamma_d^{d/2} / \ln(L_N) \right]^{2/d} \right]} \beta^{-1} (E - E_N^{1,\omega})^{-1} \, \mathrm{d} \mathcal N_N^{\omega} (E) \right] \, ud \mathds P( \omega) \\
	& \le (\beta c_3)^{-1} \lim\limits_{N \to \infty} N^{1-\eta_1} \int\limits_{\widetilde \Omega_N^{c_2, c_3}} \int\limits_{\left( E_N^{c_2,\omega},\left[ (1+ \eta_1/2) \nu \gamma_d^{d/2} / \ln(L_N) \right]^{2/d} \right]} \, \mathrm{d} \mathcal N_N^{\omega} (E) \, \ud \mathds P( \omega) \\
	& \le (\beta c_3)^{-1} \lim\limits_{N \to \infty} N^{1-\eta_1} \int\limits_{\widetilde \Omega_N^{c_2, c_3}} \mathcal N_N^{\mathrm{I},\omega} \left( \left( \dfrac{(1 + \eta_1/2) \nu \gamma_d^{d/2}}{ \ln (L_N)}\right)^{2/d} \right) \ud \mathds P( \omega) \ .
	\end{align*}
	Using the fact that $\mathcal N_N^{\mathrm{I},\omega}(E) \ge 0$ for $\mathds P$-almost all $\omega \in \Omega$, for all $E \in \mathds R$, and for all $N \in \mathds N$ as well as employing our conditions (\ref{assumption 3}) and (\ref{assumption 4}) from Assumptions~\ref{assumptions}, we get %we continue, employing Lemma~\ref{lemma E2 2} ($\mathcal E = (\nu \gamma_d^{d/2} )^{-1} \ln(M^{1/2})$) and Theorem~\ref{Lifshitz Auslaufer one dimensional 1} ($\widetilde{M}=M^{1/2}$),
	\begin{align*}
	\lim\limits_{N \to \infty} A_N^{(1)}(\epsilon) & \le (\beta c_3)^{-1} \lim\limits_{N \to \infty} N^{1-\eta_1} \E \ \mathcal N_N^{\mathrm{I},\omega} \left( \left( \dfrac{(1 + \eta_1/2) \nu \gamma_d^{d/2}}{ \ln (L_N)}\right)^{2/d} \right)  \\
	& \le (\beta c_3)^{-1} c_1 \lim\limits_{N \to \infty} N^{1-\eta_1} \mathcal N_{\infty}^{\mathrm{I}} \left( \left( \dfrac{(1 + \eta_1/2) \nu \gamma_d^{d/2}}{ \ln (L_N)}\right)^{2/d} \right) \\
	& = 0 \ .
	\end{align*}
	Secondly, for arbitrary $N \in \mathds{N}$ we conclude that for all $E \ge [(1+ \eta_1/2) \nu \gamma_d^{d/2} / \ln(L_N)]^{2/d}$ and for all $\omega \in \widetilde \Omega_N^{c_2, c_3}$,
	$$E \ge \left( \left( 1+ \dfrac{\eta_1}{2} \right) \dfrac{\nu \gamma_d^{d/2}}{\ln(L_N)} \right)^{2/d} \ge \left( \dfrac{1+\eta_1/2}{1 + \eta_1/4} \right)^{2/d} E_N^{1,\omega}$$
	and, consequently,
	$$E - \mu_N^{\omega} \ge E - E_N^{1,\omega} \ge \left[ 1 - \left( \dfrac{1 + \eta_1/4}{1 + \eta_1/2} \right)^{2/d} \right] E = c_4 E \ ,$$
	where we have set $c_4: = 1 - ( [1 + \eta_1/4]/[1 + \eta_1/2] )^{2/d} > 0$. Hence,
	\begin{align*}
	& \lim\limits_{\epsilon \to 0+} \limsup\limits_{N \to \infty} A_N^{(2)}(\epsilon)\\
	& \quad = \lim\limits_{\epsilon \to 0+} \limsup\limits_{N \to \infty} \int\limits_{\widetilde \Omega_N^{c_2, c_3}} \int\limits_{\left( \left[ (1+ \eta_1/2) \nu \gamma_d^{d/2} / \ln(L_N) \right]^{2/d},\epsilon \right]}\mathcal{B}(E - \mu_N^{\omega}) \, \mathrm{d} \mathcal N_N^{\omega} (E) \, \ud \mathds P( \omega) \\
	%\le \, & \lim\limits_{\epsilon \to 0+} \limsup\limits_{N \to \infty} \int\limits_{\Omega_2(\eta,c_2)} \left[ \, \int\limits_{\left[ (1+ \eta_1/2) \nu \gamma_d^{d/2} / \ln(L_N) \right]^{2/d}}^{\epsilon} \left( \e^{\beta (E - E_N^{1,\omega})} - 1 \right)^{-1} \mathrm{d} \mathcal N_N^{\omega} (E) \right] \, ud \mathds P( \omega) \\
	%& \le \lim\limits_{\epsilon \to 0+} \limsup\limits_{N \to \infty} \int\limits_{\widetilde \Omega_N^{c_2, c_3}} \left[ \, \int\limits_{[[ (1+ \eta_1/2) \nu \gamma_d^{d/2} / \ln(L_N) ],\epsilon]} \left( \e^{\beta (1/2) E } - 1 \right)^{-1} \mathrm{d} \mathcal N_N^{\omega} (E) \right] \, ud \mathds P( \omega) \\
	& \quad \le\lim\limits_{\epsilon \to 0+} \limsup\limits_{N \to \infty} \int\limits_{\widetilde \Omega_N^{c_2, c_3}}  \int\limits_{\left( \left[ (1+ \eta_1/2) \nu \gamma_d^{d/2} / \ln(L_N) \right]^{2/d},\epsilon \right]} \mathcal B \left( c_4 E \right)\, \mathrm{d} \mathcal N_N^{\omega} (E) \, \ud \mathds P( \omega) \\
	& \quad \le (\beta c_4)^{-1} \lim\limits_{\epsilon \to 0+} \limsup\limits_{N \to \infty} \int\limits_{\widetilde \Omega_N^{c_2, c_3}} \int\limits_{\left( \left[ (1+ \eta_1/2) \nu \gamma_d^{d/2} / \ln(L_N) \right]^{2/d},\epsilon \right]} E^{-1} \, \mathrm{d} \mathcal N_N^{\omega} (E) \, \ud \mathds P( \omega) \ .
	\end{align*}
	Employing an integration by parts (for Lebesgue--Stieltjes integrals, see, e.g., \cite[Theorem 21.67]{hewitt1965real}), we have
	\begin{align*}
	& \lim\limits_{\epsilon \to 0+} \limsup\limits_{N \to \infty} A_N^{(2)}(\epsilon)\\
%	& \quad \le \lim\limits_{\epsilon \to 0+} \limsup\limits_{N \to \infty} \int\limits_{\widetilde \Omega_N^{c_2, c_3}} \left[ \, \int\limits_{\left( \left[ (1+ \eta_1/2) \nu \gamma_d^{d/2} / \ln(L_N) \right]^{2/d},\epsilon \right]}\mathcal{B}(E - \mu_N^{\omega}) \, \mathrm{d} \mathcal N_N^{\omega} (E) \right] \, ud \mathds P( \omega) \\
	& \quad \le (\beta c_4)^{-1} \lim\limits_{\epsilon \to 0+} \limsup\limits_{N \to \infty} \int\limits_{\widetilde \Omega_N^{c_2, c_3}} \Bigg[ \epsilon^{-1} \mathcal N_N^{\mathrm{I},\omega} (2 \epsilon) \\
	& \qquad \qquad \qquad \qquad \qquad \qquad \qquad + \,
	\left. \int\limits_{\left[ (1+ \eta_1/2) \nu \gamma_d^{d/2} / \ln(L_N) \right]^{2/d}}^{\epsilon} \mathcal N_N^{\mathrm{I},\omega} (E) E^{-2} \, \mathrm{d} E \right] \, \ud \mathds P( \omega) \ .
	\end{align*}
	We then use again the fact that $\mathcal N_N^{\mathrm{I},\omega}(E) \ge 0$ for $\mathds P$-almost all $\omega \in \Omega$, all $E \in \mathds R$, and all $N \in \mathds{N}$ and employ Assumption~\ref{assumptions}~\eqref{assumption 3} and Remark~\ref{assumptions remark} to obtain
	\begin{align*}
	& \lim\limits_{\epsilon \to 0+} \limsup\limits_{N \to \infty} A_N^{(2)}(\epsilon)\\
	%\vdots \\
	%\shortintertext{ with Fubini--Tonelli theorem}
	& \quad \le (\beta c_4)^{-1} \lim\limits_{\epsilon \to 0+} \limsup\limits_{N \to \infty} \Bigg[ \epsilon^{-1} \int\limits_{\widetilde \Omega_N^{c_2, c_3}} \mathcal N_N^{\mathrm{I},\omega} ( 2 \epsilon) \, \ud \mathds P( \omega) \\
	& \qquad \qquad \qquad \qquad \qquad \qquad + \, \left. \int\limits_{\left[ (1+ \eta_1/2) \nu \gamma_d^{d/2} / \ln(L_N) \right]^{2/d}}^{\epsilon}E^{-2} \int\limits_{\widetilde \Omega_N^{c_2, c_3}} \, \mathcal N_N^{\mathrm{I},\omega} (E) \, \ud \mathds P( \omega) \,  \mathrm{d} E \right] \\
	%\shortintertext{because it is $\mathcal N_N^{\mathrm{I},\omega}(E) \ge 0$ for all $E \in \mathds R$ and all $\omega \in \Omega$}
	& \quad \le (\beta c_4)^{-1} \lim\limits_{\epsilon \to 0+} \limsup\limits_{N \to \infty} \Bigg[ \epsilon^{-1} \, \E \ \mathcal N_N^{\mathrm{I},\omega} ( 2 \epsilon) \\
	& \qquad \qquad \qquad \qquad \qquad \qquad + \, \left. \int\limits_{\left[ (1+ \eta_1/2) \nu \gamma_d^{d/2} / \ln(L_N) \right]^{2/d}}^{\epsilon}E^{-2} \, \E \ \mathcal N_N^{\mathrm{I},\omega} (E) \, \mathrm{d} E \right] \\
	%\shortintertext{due to our assumption \eqref{theorem macroscopic occupation allgemein Annahme eins}}
	& \quad \le \dfrac{c_1}{\beta c_4} \lim\limits_{\epsilon \to 0+} \left[ \epsilon^{-1} \mathcal N_{\infty}^{\mathrm{I}}(2 \epsilon)+ \int\limits_{0}^{\epsilon} \mathcal N_{\infty}^{\mathrm{I}}(E) E^{-2} \, \mathrm{d} E \right] \\
	%\shortintertext{with assumptions \eqref{Annahme gen BEC untere Grenze} and \eqref{theorem macroscopic occupation allgemein Annahme vier}}
	& \quad = 0 \ .
	\end{align*}

	In a third step we show that
	$$\lim\limits_{\epsilon \to 0+} \limsup\limits_{N \to \infty} A_N^{(3)}(\epsilon) = 0$$
	proceeding similarly as before. Note that for all $N \in \mathds{N}$, all $\omega \in \widehat \Omega_N^{c_2, c_3}$, and all $E > E_N^{c_2 + 1,\omega}$ we now have
	$$E \ge E_N^{c_2 + 1,\omega} \ge \left( \left( 1+ \dfrac{\eta_1}{2} \right) \dfrac{\nu \gamma_d^{d/2}}{\ln(L_N)} \right)^{2/d} \ge \left( \dfrac{1+\eta_1/2}{1 + \eta_1/4} \right)^{2/d} E_N^{1,\omega}$$
	and hence
	$$E - \mu_N^{\omega} \ge E - E_N^{1,\omega} \ge \left[1 - \left( \dfrac{1 + \eta_1/4}{1 + \eta_1/2} \right)^{2/d} \right] E = c_4 E \ .$$
	Consequently, 
	\begin{align*}
	& \lim\limits_{\epsilon \to 0+} \limsup\limits_{N \to \infty} A_N^{(3)}(\epsilon) %\\
	\le \dfrac{c_1}{\beta c_4} \lim\limits_{\epsilon \to 0+} \left[ \epsilon^{-1} \mathcal N_{\infty}^{\mathrm{I}}(2 \epsilon)+ \int\limits_{0}^{\epsilon} \mathcal N_{\infty}^{\mathrm{I}}(E) E^{-2} \, \mathrm{d} E \right] %\\
	%%\shortintertext{with assumptions \eqref{Annahme gen BEC untere Grenze} and \eqref{theorem macroscopic occupation allgemein Annahme vier}}
	= 0 \ .
	\end{align*}
	Summing up the last steps we have shown that $\lim_{\epsilon \to 0+} \limsup_{N \to \infty} A_N(\epsilon) = 0$.
		
	In a final step we bound the term in line~\eqref{proof BEC type ns LSM 3}: For all $\epsilon > 0$,
	\begin{align*}
	\limsup\limits_{N \to \infty} B_N(\epsilon) & \le \limsup\limits_{N \to \infty} \int\limits_{\widehat \Omega_N^{(2),\epsilon} \cap \, \Omega \backslash \Omega_N^{c_2, c_3}} \int\limits_{(0,\epsilon]} \mathcal{B}(E - \mu_N^{\omega})\, \mathrm{d} \mathcal N_N^{\omega} (E)\, \ud \mathds P( \omega) \\
	%& \le \rho \limsup\limits_{N \to \infty} \mathds P\left( \widehat \Omega_N^{(2),\epsilon} \cap \Omega \backslash \Omega_N^{c_2, c_3} \right) \\
	& \le \rho \limsup\limits_{N \to \infty} \mathds P\left( \Omega \backslash \Omega_N^{c_2, c_3} \right) \ .
	\end{align*}
	Thus,
	\begin{align*}
	& \liminf\limits_{N \to \infty} \E \int\limits_{\left( 0,E_N^{c_2,\omega} \right]} \mathcal{B}(E - \mu_N^{\omega}) \, \ud \mathcal N_N^{\omega}(E)  \ge \rho - \rho_c - \rho \limsup\limits_{N \to \infty} \mathds P\left( \Omega \backslash \Omega_N^{c_2, c_3} \right) \ .
	\end{align*}
\end{proof}
We can now state our main result which allows to conclude that, for particle densities larger than the critical one and given a certain ``gap condition'' is satisfied, one indeed has BEC of type I in probability.
\begin{theorem}[Type-I BEC in probability]\label{MainResult} Assume that Assumptions~\ref{assumptions} are fulfilled and that $\rho > \rho_{c}(\beta)$. If $H_{N,\omega}$ is such that the ``gap condition'' for some $c_2 \in \mathds{N}$ and $c_3 >0$ is fulfilled, i.e.,  
	\begin{equation} \label{gap condition}
	\lim_{N \rightarrow \infty}\mathds{P}\left(\Omega_N^{c_2, c_3}\right)=1 \ ,
	\end{equation}
then, for all $\eta > 0$,  
	\begin{equation} \label{(2.18)}
	\lim_{N \rightarrow \infty}\mathds{P}\left(\left|\frac{1}{N}\sum_{j=1}^{c_2}n^{j,\omega}_N-\frac{\rho-\rho_c(\beta)}{\rho}  \right| <\eta \right)  =1\ .
	\end{equation}
	In addition, for all $\eta > 0$ and all $j \geq c_2+1$ one has
	\begin{equation} 
	\lim_{N \rightarrow \infty}\mathds{P}\left(\frac{n^{j,\omega}}{N} \ge \eta \right)  =0
	\end{equation}
	In other words, one has type-I BEC in probability if the gap condition \eqref{gap condition} is fulfilled.
\end{theorem}
\begin{proof}
We define $\rho_0(\beta) := \rho - \rho_c(\beta)$ for convenience. Assumption \eqref{gap condition} together with Proposition~\ref{Proposition 1} imply that
	
	\begin{align}
	& \liminf\limits_{N \to \infty} \E  \int\limits_{\left( 0,E_N^{c_2,\omega} \right]} \mathcal{B}(E - \mu_N^{\omega}) \, \mathrm{d} \mathcal N_N^{\omega}(E)  \ge \rho_0(\beta) \ .
	\end{align}
		On the other hand, with Lemma~\ref{limsup 1/N sum1c2 n le rho_0/rho + (2 beta epsilon Ninfty}, the (reverse) Fatou Lemma, and the fact that
		\begin{align*}\
 \int\limits_{\left( 0,E_N^{c_2,\omega} \right]} \mathcal{B}(E - \mu_N^{\omega}) \, \mathrm{d} \mathcal N_N^{\omega}(E)	 \le \rho 
	 \end{align*}
	 for $\mathds P$-almost all $\omega \in \Omega$ and for all $N \in \mathds{N}$, we obtain
		\begin{align*}
	\limsup\limits_{N \to \infty} \E \int\limits_{\left( 0,E_N^{c_2,\omega} \right]} \mathcal{B}(E - \mu_N^{\omega}) \, \mathrm{d} \mathcal N_N^{\omega}(E)  & \le \E \limsup\limits_{N \to \infty} \int\limits_{\left( 0,E_N^{c_2,\omega} \right]} \mathcal{B}(E - \mu_N^{\omega}) \, \mathrm{d} \mathcal N_N^{\omega}(E)  \\
	& \le \rho_0(\beta) \ .
	\end{align*}
	In conclusion, remembering the relation $\rho=N/L_N$, 
	\begin{align} \label{zkjzskjh234kh}
	\lim\limits_{N \to \infty} \E  \int\limits_{\left( 0,E_N^{c_2,\omega} \right]} \mathcal{B}(E - \mu_N^{\omega}) \, \mathrm{d} \mathcal N_N^{\omega}(E) = \lim\limits_{N \to \infty} \E \ \dfrac{\rho}{N} \sum\limits_{j=1}^{c_2} n_N^{j,\omega}  = \rho_0(\beta) \ .
	\end{align}
	In addition,
	 %In conclusion,
	%\begin{align} \label{zkjzskjh234kh}
	% \lim\limits_{N \to \infty} \E \left[ \, \dfrac{\rho}{N} \sum\limits_{j =1}^{c_2} n_N^{j,\omega} \right] = 
	%	 \lim\limits_{N \to \infty} \E \left[ \, \int\limits_{\left( 0,E_N^{c_2,\omega} \right]} \mathcal{B}(E - \mu_N^{\omega}) \, \mathcal %N_N^{\omega}(\mathrm{d} E) \right] = \rho_0(\beta) \ .
	%\end{align}
	for an arbitrary $\widetilde \eta > 0$ and with Lemma~\ref{limsup 1/N sum1c2 n le rho_0/rho + (2 beta epsilon Ninfty} we obtain
	\begin{align*}
	0 & \le \limsup\limits_{N \to \infty} \int\limits_{\frac{1}{N} \sum_{j =1}^{c_2} n_N^{j,\omega} > \rho_0(\beta) / \rho} \left( \dfrac{1}{N} \sum\limits_{j =1}^{c_2} n_N^{j,\omega} - \dfrac{\rho_0(\beta)}{\rho} \right) \ud \mathds P(\omega) \\
	& \le \limsup\limits_{N \to \infty} \int\limits_{\rho_0(\beta)/ \rho < \frac{1}{N} \sum_{j =1}^{c_2} n_N^{j,\omega} < (\rho_0(\beta) / \rho) + \widetilde \eta} \left( \dfrac{1}{N} \sum\limits_{j =1}^{c_2} n_N^{j,\omega} - \dfrac{\rho_0(\beta)}{\rho} \right) \ud \mathds P(\omega) \\
	& \qquad + \,\limsup\limits_{N \to \infty} \int\limits_{\frac{1}{N} \sum_{j =1}^{c_2} n_N^{j,\omega} \ge (\rho_0(\beta) / \rho) + \widetilde \eta} \left( \dfrac{1}{N} \sum\limits_{j =1}^{c_2} n_N^{j,\omega} - \dfrac{\rho_0(\beta)}{\rho} \right) \ud \mathds P(\omega) \\
	& \le \widetilde \eta + \left( 1 - \dfrac{\rho_0(\beta)}{\rho} \right) \lim\limits_{N \to \infty} \mathds P \left( \frac{1}{N} \sum_{j =1}^{c_2} n_N^{j,\omega} \ge \dfrac{\rho_0(\beta)}{\rho} + \widetilde \eta \right) \\
	& \le \widetilde \eta \ .
	\end{align*}
	Thus, by \eqref{zkjzskjh234kh}, 
	\begin{align*}
	 0 & = \lim\limits_{N \to \infty} \E \left[ \dfrac{1}{N} \sum\limits_{j =1}^{c_2} n_N^{j,\omega} - \dfrac{\rho_0(\beta)}{\rho} \right] = \lim\limits_{N \to \infty} \int\limits_{\Omega} \left( \dfrac{1}{N} \sum\limits_{j =1}^{c_2} n_N^{j,\omega} - \dfrac{\rho_0(\beta)}{\rho} \right) \ud \mathds P(\omega) \\
	 %& = \lim\limits_{N \to \infty} \int\limits_{\frac{1}{N} \sum_{j =1}^{c_2} n_N^{j,\omega} \le \rho_0(\beta) / \rho} \left( \dfrac{1}{N}\sum\limits_{j =1}^{c_2} n_N^{j,\omega} - \dfrac{\rho_0(\beta)}{\rho} \right) \ud \mathds P(\omega) \\
	 %& \qquad + \, \lim\limits_{N \to \infty} \int\limits_{\frac{1}{N} \sum_{j =1}^{c_2} n_N^{j,\omega} > \rho_0(\beta) / \rho} \left( \dfrac{1}{N} \sum\limits_{j =1}^{c_2} n_N^{j,\omega} - \dfrac{\rho_0(\beta)}{\rho} \right) \ud \mathds P(\omega)	 \\
	 & = \lim\limits_{N \to \infty} \int\limits_{\frac{1}{N} \sum_{j =1}^{c_2} n_N^{j,\omega} \le \rho_0(\beta) / \rho} \left( \dfrac{1}{N} \sum\limits_{j =1}^{c_2} n_N^{j,\omega} - \dfrac{\rho_0(\beta)}{\rho} \right) \ud \mathds P(\omega) \ .
	\end{align*}
	We conclude
	\begin{align*}
	 \lim\limits_{N \to \infty} \E \left| \dfrac{1}{N} \sum\limits_{j =1}^{c_2} n_N^{j,\omega} - \dfrac{\rho_0(\beta)}{\rho} \right|  %\\
	 %& \quad = \lim\limits_{N \to \infty} \int\limits_{\frac{1}{N} \sum_{j =1}^{c_2} n_N^{j,\omega} \le \rho_0(\beta) / \rho} \left( \dfrac{\rho_0(\beta)}{\rho} - \dfrac{1}{N} \sum\limits_{j =1}^{c_2} n_N^{j,\omega} \right) \\
	 %& \qquad + \, \lim\limits_{N \to \infty} \int\limits_{\frac{1}{N} \sum_{j =1}^{c_2} n_N^{j,\omega} > \rho_0(\beta) / \rho} \left( \dfrac{1}{N} \sum\limits_{j =1}^{c_2} n_N^{j,\omega} - \dfrac{\rho_0(\beta)}{\rho} \right) \\
	 = 0
	\end{align*}
    and, consequently, for all $\eta > 0$,
    \begin{align*}
		 \lim\limits_{N \to \infty} \mathds P \left( \left| \dfrac{1}{N} \sum\limits_{j = 1}^{c_2} n_N^{j,\omega} - \dfrac{\rho_0(\beta)}{\rho} \right| < \eta \right) = 1 \ .
	\end{align*}
	Finally, we show the last part of the theorem: Using the fact that $N^{-1} \sum_{j=1}^{c_2 + 1} n_N^{j,\omega} \le 1$ for $\mathds P$-almost all $\omega \in \Omega$ and for all $N \in \mathds{N}$, Lemma~\ref{limsup 1/N sum1c2 n le rho_0/rho + (2 beta epsilon Ninfty}, and the (reverse) Fatou Lemma we obtain
		 \begin{align*}
		\limsup\limits_{N \to \infty} \mathds E \ \dfrac{1}{N} \sum\limits_{j=1}^{c_2 + 1} n_N^{j,\omega}  \le \dfrac{\rho_0(\beta)}{\rho} \ .
		\end{align*}
	In addition,
	\begin{align*}
	 \liminf\limits_{N \to \infty} \E \ \dfrac{1}{N} \sum\limits_{j =1}^{c_2 + 1} n_N^{j,\omega}  & \ge \lim\limits_{N \to \infty} \E \ \dfrac{1}{N} \sum\limits_{j =1}^{c_2} n_N^{j,\omega} = \dfrac{\rho_0(\beta)}{\rho} \ ,
	 \end{align*}
	 see \eqref{zkjzskjh234kh}. Therefore,
	 \begin{align*}
	 \dfrac{\rho_0(\beta)}{\rho} & = \lim\limits_{N \to \infty} \E \ \dfrac{1}{N} \sum\limits_{j =1}^{c_2 + 1} n_N^{j,\omega}  %= \lim\limits_{N \to \infty} \E \left[ \dfrac{1}{N} \sum\limits_{j =1}^{c_2} n_N^{j,\omega} \right] + \lim\limits_{N \to \infty} \E \left[ \dfrac{1}{N} n_N^{c_2 + 1,\omega} \right] \\
	 = \dfrac{\rho_0(\beta)}{\rho} + \lim\limits_{N \to \infty} \E \ \dfrac{n_N^{c_2 + 1,\omega}}{N}   \ .
	\end{align*}
	%Hence, we conclude
	%\begin{align*}
	 %\lim\limits_{N \to \infty} \mathds E \left[ \dfrac{1}{N} n_N^{c_2 + 1,\omega} \right] = \lim\limits_{N \to \infty} \mathds E \left[ \left| %\dfrac{1}{N} n_N^{c_2 + 1,\omega} \right| \right] = 0 \ .
	%\end{align*}
 For any $j \ge c_2 + 1$ we thus conclude, due to the fact that $n_N^{j,\omega} \le n_N^{c_2 + 1,\omega}$ for $\mathds P$-almost all $\omega \in \Omega$ and for all $N \in \mathds{N}$,
 \begin{align*}
	 \lim\limits_{N \to \infty} \mathds E  \left| \dfrac{n_N^{j,\omega}}{N} \right| = 0
	\end{align*}
	and, in particular, for all $\eta > 0$, 
	\begin{align*}
	 \lim\limits_{N \to \infty} \mathds P\left( \dfrac{n_N^{j,\omega}}{N} < \eta \right) = 1 \ .
	\end{align*}
\end{proof}

%\begin{rem}
%	\textcolor{red}{Wieso folgt aus \eqref{(2.18)} type-I BEC? Genauer erklären}
%\end{rem}
%
\begin{remark}\label{RemarkXXX} Since $c_2 \, n_N^{1,\omega} \ge \sum_{j=1}^{c_2} n_N^{j,\omega}$ for $\mathds P$-almost all $\omega \in \Omega$ and for all $N \in \mathds{N}$, Theorem~\ref{MainResult} immediately implies that, for $\rho > \rho_c$ and for all $\eta > 0$, 
		\begin{align*}
		\lim\limits_{N \to \infty} \mathds P \left( \dfrac{1}{N} n_N^{1,\omega} > c_2^{-1} \dfrac{\rho - \rho_c(\beta)}{\rho} - \eta \right) %\ge \lim\limits_{N \to \infty} \mathds P \left( \dfrac{1}{N} \sum\limits_{j = 1}^{c_2} n_N^{j,\omega} \ge \dfrac{\rho_0(\beta)}{\rho} - c_2 \eta \right)
		= 1\ .
		\end{align*}
		This shows that the ground state is macroscopically occupied in probability. 
	\end{remark}
In a next result we establish another consequence of Theorem~\ref{MainResult}, namely, existence of BEC in the $rth$ mean. 
\begin{cor}[Type-I BEC in the $r$th mean]\label{Corollary type-I BEC in rth mean} Assume that $\rho > \rho_{c}(\beta)$ and that the gap condition \eqref{gap condition} is fulfilled. Then, for all $r \ge 1$, one has
	\begin{equation}
	\lim_{N \rightarrow \infty}\mathds{E}\left|\frac{1}{N}\sum_{j=1}^{c_2}n^{j,\omega}_N-\frac{\rho-\rho_c(\beta)}{\rho}   \right|^r  =0
	\end{equation}
	as well as
	\begin{equation}
	\lim_{N \rightarrow \infty}\mathds{E}\left|\frac{n^{j,\omega}_N}{N} \right|^r  =0
	\end{equation}
	for every $j \geq c_2$.
\end{cor}
\begin{proof} The proof readily follows from Theorem~\ref{MainResult} taking into account standard results from probability theory. 
	Namely, from Theorem~\ref{MainResult} we conclude that the two sequences of random variables 
	\begin{equation*}
	\left(\frac{1}{N}\sum_{j=1}^{c_2}n^{j,\omega}_N-\frac{\rho-\rho_c(\beta)}{\rho}    \right)_{N \in \mathds{N}} \ \text{and}\ \left(\frac{n^{c_2+1,\omega}_N}{N} \right)_{N \in \mathds{N}}
	\end{equation*}
	converge to zero in probability. In addition, since 
	\begin{equation*}
	-1 \leq \frac{1}{N}\sum_{j=1}^{c_2}n^{j,\omega}_N-\frac{\rho-\rho_c(\beta)}{\rho} \leq 1 
	\end{equation*}
	and
	\begin{equation*}
	0 \leq \frac{n^{c_2+1,\omega}_N}{N} \leq 1
	\end{equation*}
	for $\mathds P$-almost all $\omega \in \Omega$, both sequences are $\mathds P$-almost surely uniformly bounded which then implies the statement. 
\end{proof}

\begin{cor}[Almost sure macroscopic occupation of the ground state]\label{MacroscopicOccupationGroundstate}
	If the requirements of Theorem~\ref{MainResult} are fulfilled, then the ground state is $\mathds P$-almost surely macroscopically occupied,
		\begin{align*}
		\mathds P \left( \limsup\limits_{N \to \infty} \dfrac{n_N^{1,\omega}}{N} > 0 \right) = 1 \ .
		\end{align*}
		%see Definition \ref{Definition makroskopische Besetzung}. 
		%
	\end{cor}
	\begin{proof} The statement can be proved in a similar vein as \cite[Theorem 3.5]{KPS182}, taking Remark~\ref{RemarkXXX} into account. 
		
		Alternatively, \eqref{(2.18)} implies, for $\mathds P$-almost all $\omega \in \Omega$, the existence of a subsequence $\left(\frac{1}{N_k}\sum_{j=1}^{c_2}n^{j,\omega}_{N_k}\right)_{k\in \mathds N}$ converging to $\frac{\rho-\rho_c(\beta)}{\rho}$ in the limit $k \rightarrow \infty$. Consequently,
		\begin{equation*}
	\frac{\rho-\rho_c(\beta)}{\rho} \leq 	\limsup_{N \rightarrow \infty} \frac{1}{N}\sum_{j=1}^{c_2}n^{j,\omega}_{N} \leq c_2 \limsup\limits_{N \to \infty} \dfrac{n_N^{1,\omega}}{N}
		\end{equation*}
		$\mathds P$-almost surely, which implies the statement. 
	\end{proof}
\begin{remark} By the reverse triangle inequality we also see the following: Given the requirements in Theorem~\ref{MainResult} are met, then
	$$\lim\limits_{N \to \infty} \E \left( \dfrac{1}{N} \sum\limits_{j=1}^{c_2} n_N^{j,\omega} \right)^r  = \E \left( \dfrac{\rho - \rho_c(\beta)}{\rho} \right)^r $$
	for all $r \ge 1$. In particular,
	\begin{align*} 
	\liminf_{N\to\infty} \mathds E \ \dfrac{n_N^{1,\omega}}{N} \ge \dfrac{1}{c_2} \dfrac{\rho - \rho_c(\beta)}{\rho} \ ,
	\end{align*}
	i.e., the ground state is macroscopically occupied in expectation.
	\end{remark}
	%
%For details regarding the proof of the last two statements, see the last part of Section 3 in \cite{KPS182}.

	\begin{remark} \label{remark Poisson erfuellt bedinungen}
		A Poisson random potential on $\mathds R^d$, $d \in \mathds{N}$, with a single-impurity potential that is a non-negative, compactly supported, and bounded function fulfills all requirements of Assumptions \ref{assumptions}, see, for example, \cite[Theorem 4.6]{sznitman1998brownian},\cite[Theorem 5.20, 5.25, and Theorem 10.2]{pastur1992spectra}, as well as \cite{leschke2003survey}.
		%We show that Assumptions \ref{assumptions}~\eqref{assumption 3} are met in Lemma~\ref{lemma E2} in more details.
		Theorem~\ref{MainResult} and Corollaries~\ref{Corollary type-I BEC in rth mean} and \ref{MacroscopicOccupationGroundstate} thus apply to such a  Poisson random potential. To the best of our knowledge, however, it is so far not known whether the ``gap condition'' of Theorem~\ref{MainResult} is then valid. In the next chapter, we study a special case of Poisson random potentials on $\mathds R$ for which we are able to confirm the ``gap condition.''
	\end{remark}

\section{An example: the Luttinger--Sy model} \label{section Luttinger Sy infinite strength}

In this section we are concerned with BEC in the Luttinger--Sy model (LS-model) \cite{luttinger1973bose,luttinger1973low} for which all the assumptions on the Hamiltonian $H_{N,\omega}$, see Assumptions~\ref{assumptions} and before, are fulfilled. It is our aim to show that also the ``gap condition'' as formulated in Theorem~\ref{MainResult} is realized in the LS-model. Hence, the corresponding versions of Theorem~\ref{MainResult} and Corollaries~\ref{Corollary type-I BEC in rth mean} and \ref{MacroscopicOccupationGroundstate} hold for the LS-model, too.

 The Luttinger--Sy model is a random one-dimensional model which is obtained by dissecting the real line $\mathds{R}$ into a ($\mathds P$-almost surely) countable number of intervals via a set of points $\{\hat x_j(\omega)\}_{j}$ generated by a Poisson point process of intensity $\nu > 0$. These points can $\mathds P$-almost surely be labeled by $\mathds Z$ such that
 $$\ldots < \hat x_{-1}(\omega) < \hat x_0(\omega) < 0 < \hat x_1(\omega) < \hat x_2(\omega) < \ldots \ .$$ We refer to \cite[Chapter 4]{kingman1993poisson} for more details regarding the Poisson point process on $\mathds R$. At each point $\hat x_j(\omega)$ one then imposes a Dirichlet boundary condition which is informally equivalent to saying that one places a Dirac-$\delta$ potential of infinite strength $\gamma=\infty$ at each Poisson point $\hat x_j(\omega)$. Hence, the one-particle Hamiltonian in the LS-model is informally given by
\begin{equation}\label{OneParticleHamiltonian}
H_{\omega}=-\frac{\ud^2}{\ud x^2}+\gamma\sum_{j\in \mathds{Z}}\delta(x-\hat x_j(\omega))\ , \quad  \omega \in \Omega\ ,\quad \gamma=\infty\ ,
\end{equation}
$(\Omega,\mathscr A,\mathds{P})$ denoting the underlying probability space. BEC is investigated employing a thermodynamic limit which makes it necessary - in a first step - to restrict $H_{\omega}$ to finite volume. For this one introduces the window $\Lambda_N:=\left(-L_N/2,+L_N/2\right) \subset \mathds{R}$ where the length $L_N$ is determined through the relation $L_N:=N/\rho$, $N \in \mathds{N}$ denotes the number of particles, and $\rho > 0$ is the particle density. On $L^2(\Lambda_N)$ one then introduces the (informal) finite-volume one-particle Hamiltonian 
\begin{equation}
H_{N,\omega}:=-\frac{\ud^2}{\ud x^2}+\gamma\sum_{j\in \mathds{Z}: \hat x_j(\omega) \in \Lambda_N}\delta(x-\hat x_j(\omega))\ , \quad  \omega \in \Omega\ ,\quad \gamma=\infty\  .
\end{equation}
We refer to \cite{LenobleZagrebnovLuttingerSy,KPS18} for a rigorous realization of $H_{N,\omega}$ as a self-adjoint operator.  

Since one imposes Dirichlet boundary conditions at each point $\hat x_j(\omega)$, the window $\Lambda_N$ effectively consists of a collection of smaller intervals of lengths $l_N^{j,\omega}:=|(\hat x_j(\omega), \hat x_{j+1}(\omega)) \cap \Lambda_N|$ on each of which the operator $H_{N,\omega}$ acts as the standard one-dimensional Dirichlet Laplacian. As a consequence, the eigenvalues of $H_{N,\omega}$ are just the collection of all the eigenvalues coming from these Dirichlet Laplacians; their eigenvalues, on the other hand, are explicitly given: namely, the $n$th eigenvalue on the interval with length $l_N^{j,\omega}$ simply is $(\pi n/l_N^{j,\omega})^2$, $n \in \mathds{N}$.

We denote the intervals generated by the Poisson point process within $\Lambda_N$ by
$$I_N^{j,\omega} := (\hat x_j(\omega), \hat x_{j+1}(\omega)) \cap \Lambda_N$$
for all $j \in \mathds Z$ and all $N \in \mathds{N}$. Also, let $l^{k,\omega}_{N,>}$, $k \in \mathds{N}$, be the $k$th largest length of these intervals. We denote by $\kappa_N^{\omega}$ the number of Poisson points within $\Lambda_N$.
Consequently, for $\mathds P$-almost all $\omega \in \Omega$ and for all $N \in \mathds{N}$ there are $\kappa_N^{\omega} + 1$ many lengths $l_N^{j,\omega}$ that are larger than zero. Furthermore, $\{ \ljomega : j \in \mathds Z \backslash \{0\} \}\cup \{|(\hat x_0(\omega),0)|\} \cup \{|(0,\hat x_1(\omega))|\}$ where $\ljomega := |(\hat x_j(\omega), \hat x_{j+1}(\omega)|$ are mutually independent, exponentially distributed random variables with common probability density $\nu \e^{-\nu l}$ \cite[Section 4.1]{kingman1993poisson}. 

Another main advantage of considering the LS-model is that there is an explicit expression for the integrated density of states: One has
\begin{equation}
\mathcal{N}^{\mathrm I}_{\infty}(E)=\nu\frac{\mathrm{e}^{-\nu\pi E^{-1/2}}}{1-\mathrm{e}^{-\nu\pi E^{-1/2}}} \mathds 1_{(0,\infty)}(E)\ , \quad E \in \mathds R\ ,
\end{equation}
see \cite[Proposition 3.2]{lenoble2004bose} and also \cite{EggarterSomeExact, luttinger1973bose, luttinger1973low}. 

Now we state a version of Theorem~\ref{MainResult} for the LS-model. For this, recall the  ``gap condition'' \eqref{Definition Omega4}. 
\begin{theorem}[Gap condition in the LS-model] \label{Theorem gap condition LSmodel} In the LS-model one has 
	$$ \lim\limits_{N \to \infty} \mathds P(\Omega_{N}^{1,1}) = 1 \ .$$
\end{theorem}
\begin{proof}
%	For convenience, we set $\rho_0(\beta) := \rho - \rho_{c}(\beta)$. We use Corollary~\ref{Corollary type-I BEC in rth mean} (with $c_2 = 4$) for the proof of this theorem. Consequently, we need to show the energy gap condition
%	$$ \lim\limits_{N \to \infty} \mathds P(\Omega_{N}^{4,1}) = 1 $$
%	for an $0<\eta_1<1$, where
%	$$ \Omega_{N}^{4,1}:=\left\{\omega \in \Omega: E^{5,\omega}_N-E^{1,\omega}_N \geq N^{-1+\eta_1}\ \text{and} \ E^{1,\omega}_N \leq \left[\left(1+\frac{\eta_1}{4}\right)\frac{\nu \pi}{\ln L_N} \right]^{2} \right\} \ ,$$   
%	see~\eqref{Definition Omega4} and~\eqref{gap condition}. Also, note that $\gamma_1 = \pi^2$. 

We will prove that the gap condition is actually fulfilled for any value $0 < \eta_1 < 1/2$ of the constant $\eta_1$ appearing in Assumptions~\ref{assumptions}~\eqref{assumption 4} and definition \eqref{Definition Omega4}.
	
	Let $0 < \eta_1 < 1/2$ be arbitrary. We conclude with Lemma~\ref{lemma llargest to infty almost surely} that
	\begin{align*}
	\mathds P \left( \lim\limits_{N \to \infty} \left[ \left( 1 + \dfrac{\eta_1}{3} \right)^{-1} \dfrac{\ln(L_N)}{\nu} - \llargest \right]_+ = 0 \right) = 1
	\end{align*}
	and consequently, for an arbitrary $\eta > 0$, 
	\begin{align*}
	\lim\limits_{N \to \infty} \mathds P \left( \llargest \le \left( 1 + \dfrac{\eta_1}{2} \right)^{-1} \dfrac{\ln(L_N)}{\nu} \right) & \le \lim\limits_{N \to \infty} \mathds P \left(  \left[ \left( 1 +\dfrac{\eta_1}{3} \right)^{-1} \dfrac{\ln(L_N)}{\nu} - \llargest \right]_+  > \eta \right) \\
	& = 0
	\end{align*}
	Therefore, since $E_N^{1,\omega}=\pi^2/\left(\llargest\right)^2 $, 
	\begin{align*}
	\lim\limits_{N \to \infty} \mathds P \left( E_N^{1,\omega} \le \left[ \left( 1 + \dfrac{\eta_1}{2} \right) \dfrac{\nu \pi}{\ln(L_N)} \right]^2 \right) = 1\ .
	\end{align*}
	The show the second part of the gap condition, we introduce a few definitions first: We denote by $\llargesttilde$ the largest and by $\llargesttildezwei$ the second largest length of all intervals $\{I_N^{j,\omega}\}_{j \in \mathds Z\backslash\{0\}}$ without the two outer intervals
$$I_N^{\text{R},\omega} := I_N^{j_N^{\text{max},\omega},\omega} \quad \text{ where } \quad j_N^{\text{max},\omega} := \max\left\{ j \in \mathds Z: I_N^{j,\omega} > 0 \right\}\ ,$$
and
$$I_N^{\text{L},\omega} := I_N^{j_N^{\text{min},\omega},\omega} \quad \text{ where } \quad j_N^{\text{min},\omega} := \min\left\{ j \in \mathds Z : I_N^{j,\omega} > 0 \right\}\ .$$
as well as without $I_N^{0,\omega}$ for all $N \in \mathds{N}$ and all $\omega \in \Omega$ for which they exist, that is, for which it is $\kappa_N^{\omega} \ge 4$; otherwise, i.e., if $\kappa_N^{\omega} \le 4$, we set $\llargesttilde := 0$ and $\llargesttildezwei := 0$. Furthermore, we define the sets
	$$\Omega_N^{(\text{R})} := \left\{ \omega \in \Omega : |I_N^{\text{R},\omega}| \le \dfrac{1}{2\nu} \ln (N) \right\} \ ,$$
	$$\Omega_N^{(\text{L})} := \left\{ \omega \in \Omega : |I_N^{\text{L},\omega}| \le \dfrac{1}{2\nu} \ln (N) \right\} \ ,$$
	$$\Omega_N^{(0)} := \left\{ \omega \in \Omega : |I_N^{0,\omega}| \le \dfrac{1}{2\nu} \ln (N) \right\} \ .$$
	Note that for any $N \in \mathds N$,
	$$\mathds P\left( \Omega_N^{(\text{R})} \right) \ge 1 - \mathrm{e}^{- (1/2) \ln (N)} \ ,$$
	$$\mathds P\left( \Omega_N^{(\text{L})} \right) \ge 1 - \mathrm{e}^{- (1/2) \ln (N)} \ ,$$
	see, e.g., \cite[Section 4.1]{kingman1993poisson}, and
	\begin{align*}
	\mathds P\left( \Omega_N^{(0)} \right) & \ge \mathds P \left( \left\{ \omega \in \Omega: |(0, \hat x_1(\omega))| \le \dfrac{1}{4\nu} \ln(N) \right\} \cap \left\{ \omega \in \Omega: |(\hat x_0(\omega) ,0 )| \le \dfrac{1}{4\nu} \ln(N) \right\} \right) \\
	& \ge \left( 1 - \mathrm{e}^{- (1/4) \ln (N)} \right) + \left( 1 - \mathrm{e}^{- (1/4) \ln (N)} \right) -1 = 1 - 2 \mathrm{e}^{- (1/4) \ln (N)} \ .
	\end{align*}
	With Lemma~\ref{lemma llargest to infty almost surely}, with the fact that $\mathds P\left(\limsup_{N \to \infty} (\llargest / \ln(N)\right) \le \kappa/ \nu) = 1$ for all $\kappa > 2$, which can be shown in the same way as \cite[Lemma A.2]{KPS182}, and with Lemma~\ref{Theorem Abstand grosstes und viertgrosstes Intervall}, where we choose $a = 1$ and $\eta_1 < \hat \eta < 1/2$, we obtain 
	\begin{align*}
	& \lim\limits_{N \to \infty} \mathds P \left( \left( \dfrac{\pi}{\llargestzwei} \right)^2 - \left( \dfrac{\pi}{\llargest} \right)^2 > \dfrac{1}{N^{1-\eta_1}} \right) \\
	& \quad = \lim\limits_{N \to \infty} \mathds P \left( \llargest - \llargestzwei > \dfrac{\pi^{-2}}{N^{1-\eta_1}} \cdot \dfrac{\left( \llargest \llargestzwei \right)^2}{\llargest + \llargestzwei} \right) \\
	%& \quad \ge \lim\limits_{N \to \infty} \mathds P \left( \llargest - \llargestzwei > \dfrac{ 1}{N^{1-\eta_1}} \dfrac{\left( \llargest \llargestzwei \right)^2}{\llargest} \right) \\
	%& \quad \ge \lim\limits_{N \to \infty} \mathds P \left( \llargest - \llargestzwei > \dfrac{(\nu \pi)^{-2}}{N^{1-\eta_1}} \llargest \left( \llargestzwei \right)^2 \right) \\
	& \quad \ge \lim\limits_{N \to \infty} \mathds P \left( \llargest - \llargestzwei > \dfrac{\pi^{-2}}{N^{1-\eta_1}} \left( \llargest \right)^3 \right) \\
	& \quad \ge \lim\limits_{N \to \infty} \mathds P \left( \left\{ \omega \in \Omega :  \llargest - \llargestzwei > \dfrac{ \big( 3 \nu^{-1} \ln(N) \big)^3}{\pi^2 N^{1-\eta_1}} \right\} \right. \\
	& \qquad \qquad \qquad \quad \left. \cap \left\{ \omega \in \Omega : \dfrac{3}{4} \nu^{-1} \ln(N) \le \llargest \le 3 \nu^{-1} \ln(N) \right\}\right) \\
	& \quad \ge \lim\limits_{N \to \infty} \mathds P \left( \left\{ \omega \in \Omega : \llargesttilde - \llargesttildezwei > \dfrac{ \big( 3 \nu^{-1} \ln(N) \big)^3}{\pi^2 N^{1-\eta_1}} \right\} \right. \\
	& \qquad \qquad \qquad \quad \left. \cap \left\{ \omega \in \Omega : \dfrac{3}{4} \nu^{-1} \ln(N) \le \llargest \le 3 \nu^{-1} \ln(N) \right\} \cap \Omega_N^{(0)} \cap \Omega_N^{(\text{L})} \cap \Omega_N^{(\text{R})} \right) \\
	%& \quad \ge \lim\limits_{N \to \infty} \mathds P \left( \llargest - \llargestzwei > \dfrac{ \left( 3 \nu^{-1} \ln(N) \right)^3}{\pi^2 N^{1-\eta_1}} \right) + \mathds P \left( \llargest \le 3 \nu^{-1} \ln(N) \right) - 1\\
	& \quad \ge \lim\limits_{N \to \infty} \mathds P \left( \llargesttilde - \llargesttildezwei > \dfrac{ 1}{N^{1-\hat \eta}} \right) \\  & \qquad \qquad \qquad + \, \lim\limits_{N \to \infty}  \mathds P \left( \left\{ \dfrac{3}{4} \nu^{-1} \ln(N) \le \llargest \le 3 \nu^{-1}\ln(N) \right\} \cap \Omega_N^{(0)} \cap \Omega_N^{(\text{L})} \cap \Omega_N^{(\text{R})} \right) - 1\\
	& \quad = 1 \ .
	\end{align*}
	On the other hand, for the energy gap between the ground-state energy and the first excited energy on the largest interval we have
	$$ \lim\limits_{N \to \infty} \mathds P \left( \left( \dfrac{2 \pi}{\llargest} \right)^2 - \left( \dfrac{\pi}{\llargest} \right)^2 \ge \dfrac{1}{3} \left( \dfrac{\nu \pi}{\ln(N)} \right)^2 \right) = 1 \ ,$$
	since $\lim_{N \to \infty} \mathds P(\llargest \le (3/\nu) \ln(N)) = 1$.
	
	Finally, since the second eigenvalue $E_N^{2,\omega}$ is either the energy of the first excited state on the largest interval $\llargest$ or the ground state energy of the second largest interval $\llargestzwei$, we obtain 
	\begin{align*}
	& \lim\limits_{N \to \infty} \mathds P \left( \Omega_{N}^{1,1} \right) \\
	& \quad \ge \lim\limits_{N \to \infty} \mathds P \left( E_N^{2,\omega} - E_N^{1,\omega} > \dfrac{1}{N^{1-\eta_1}} \right) + \lim\limits_{N \to \infty} \mathds P \left( E_N^{1,\omega} \le \left[ \left( 1 + \dfrac{\eta_1}{2} \right) \dfrac{\nu \pi}{\ln(L_N)} \right]^2 \right) - 1\\
	& \quad = 1 \ .
	\end{align*}
\end{proof}

\begin{remark} Theorem~\ref{Theorem gap condition LSmodel} shows that, in probability, only the ground state is macroscopically occupied in the LS-model. We also recall that Corollary~\ref{MacroscopicOccupationGroundstate} shows that the ground state is $\mathds P$-almost surely macroscopically occupied.
\end{remark}

%\subsection*{Acknowledgement}{} It is our great pleasure to thank A.-S.~Sznitman, J.~Yngvason as well as V.~Zagrebnov for interesting discussions and remarks. 

%\vspace*{0.5cm}

%\appendix
%\input{appendix}

%\newpage
\appendix

\section{Miscellaneous results} \label{Miscellaneous results}
In this appendix we collect various results, some of which we referred to in the previous text. Lemma~\ref{Lemma 2 Gen BEC} can be proved similarly as~\cite[Lemma A.7]{KPS182}.
	\begin{lemma} \label{Lemma 2 Gen BEC}
	 Under the assumptions of Theorem~\ref{TheoremGENBEC} one has, for $\rho > \rho_{c}(\beta)$ and given $\epsilon > 0$, $\mathds P$-almost surely
		\begin{align}
		\limsup\limits_{N \to \infty} \int\limits_{(\epsilon,\infty)} \mathcal B(E - \mu_N^{\omega}) \,\mathrm{d} \mathcal N_{N}^{\omega} (E) & \le \int\limits_{(\epsilon,\infty)} \mathcal B(E) \, \mathrm{d} \mathcal N_{\infty} ( E) + \dfrac{2}{\beta \epsilon} \mathcal N_{\infty}^{\mathrm{I}}(\epsilon)\ , \label{Lemma 2 Gen BEC HG 1} \\
		\liminf\limits_{N \to \infty} \int\limits_{(\epsilon,\infty)} \mathcal B(E - \mu_N^{\omega}) \,\mathrm{d} \mathcal N_{N}^{\omega} (E) & \ge \int\limits_{(\epsilon,\infty)} \mathcal B(E) \, \mathrm{d} \mathcal N_{\infty} ( E) - \dfrac{4}{\beta \epsilon} \mathcal N_{\infty}^{\mathrm{I}}(2\epsilon) \ . \label{Lemma 2 Gen BEC HG 2}
		\end{align}
	\end{lemma}

\begin{lemma} \label{limsup 1/N sum1c2 n le rho_0/rho + (2 beta epsilon Ninfty}
 If the particle density $\rho$ is larger than the critical density, $\rho > \rho_c(\beta)$, then for any $c \in \mathds{N}$ we $\mathds P$-almost surely have
 	\begin{align*} 
 \limsup\limits_{N \to \infty} \frac{1}{N} \sum_{j =1}^{c} n_N^{j,\omega} \le \dfrac{\rho_0(\beta)}{\rho} \ .
	\end{align*}
 \end{lemma}
 \begin{proof}
 Let $\epsilon > 0$ be arbitrary. We recall the well-known fact that  $(E_N^{j,\omega})_{N \in \mathds{N}}$ is, for every fixed $j \in \mathds{N}$, a monotonically decreasing sequence. Hence, for all $1 < j \leq c$ the sequence $(E_N^{j,\omega})_{N \in \mathds{N}}$ either converges to a constant $b_j > 0$ or to zero. If the sequence does not converge to zero one obtains
 $$\lim\limits_{N \to \infty} \dfrac{n_N^{j,\omega}}{N} \le \beta^{-1} \lim\limits_{N \to \infty} \dfrac{ (E_N^{j,\omega} - \mu_N^{\omega})^{-1}}{N} \le \beta^{-1} \lim\limits_{N \to \infty} \dfrac{ 2 b_j^{-1} }{N} = 0 \ ,$$
 where we used the fact that $(\mu_N^{\omega})_{N \in \mathds{N}}$ also converges to zero, see Theorem~\ref{TheoremGENBEC}.
 Thus, in either case,
 \begin{align*}
 \limsup\limits_{N \to \infty} \frac{1}{N} \sum_{j =1}^{c} n_N^{j,\omega} & \le \limsup\limits_{N \to \infty} \frac{1}{N} \sum_{j \in \mathds{N} : E_N^{j,\omega} \le \epsilon} n_N^{j,\omega} \\
 %& = \rho^{-1} \limsup\limits_{N \to \infty} \int\limits_{\left( 0,\epsilon \right]} \mathcal B(E - \mu_N^{\omega}) \, \mathrm{d} \mathcal N_N^{\omega}(E) \\
 & = \rho^{-1} \left[ \rho - \liminf\limits_{N \to \infty} \int\limits_{\left( \epsilon, \infty \right)} \mathcal B(E - \mu_N^{\omega}) \, \mathrm{d} \mathcal N_N^{\omega}(E) \right] \ .
 %& \le \rho^{-1} \left[ \rho - \int\limits_{\left( \epsilon, \infty \right)} \mathcal B(E) \, \mathcal N_{\infty}(\mathrm{d} E) + \dfrac{4}{\beta \epsilon} \mathcal N_{\infty}^{\mathrm{I}}(2\epsilon) \right] \ ,
	\end{align*}
Finally, with Lemma~\ref{Lemma 2 Gen BEC} and taking $\inf_{\epsilon}$ on both sides we arrive at
	\begin{align*} 
 \limsup\limits_{N \to \infty} \frac{1}{N} \sum_{j =1}^{c} n_N^{j,\omega} \le \dfrac{\rho - \rho_c(\beta)}{\rho}=\frac{\rho_0(\beta)}{\rho} \ ,
	\end{align*}
see also Remark~\ref{assumptions remark}.
 \end{proof}

 In the rest of this appendix, we are concerned with the Luttinger--Sy model, see Section~\ref{section Luttinger Sy infinite strength}. Recall that
 $(\hat x_j(\omega))_{j\in \mathds Z}$ with $\ldots < \hat x_{-1}(\omega) < \hat x_0(\omega) < 0 < \hat x_1(\omega) < \hat x_2(\omega) < \ldots$ are the points generated by a Poisson point process of intensity $\nu > 0$ and $\kappa_N^{\omega}$ is the number of Poisson points within the window $(-L_N/2,L_N/2)$. Furthermore, $\ljomega = |(\hat x_j(\omega), \hat x_{j+1}(\omega)|$ and $l_N^{j,\omega} =|(\hat x_j(\omega), \hat x_{j+1}(\omega)) \cap \Lambda_N|$. For the convenience of the reader, we present the next two lemmata with a proof, although more general versions of them can be found in \cite[Appendix C]{KPS18}.

	\begin{lemma} \label{Theorem k LN to nu}
For all $\epsilon > 0$ and for $\mathds P$-almost all $\omega \in \Omega$ there exists an $\widetilde N \in \mathds{N}$ such that for all $N \ge \widetilde N$ 
 \begin{align}\label{first part of Theorem k LN to nu}
(1 - \epsilon) \nu L_N < \kappa_N^{\omega} < (1 + \epsilon) \nu L_N \ . 
 \end{align}
 In particular, we $\mathds P$-almost surely have
 \begin{align*}
 \lim_{N \to \infty} \dfrac{\kappa_N^{\omega}}{L_N} = \nu \ . 
 \end{align*}
\end{lemma}
\begin{proof}
Firstly, we note that $1 - \theta + \theta \ln(\theta) > 0$ for any $\theta \in (0,\infty) \backslash \{1\}$ as well as that for arbitrary $N \in \mathds{N}$,
	\begin{align} \label{zxkjzxckhj 1}
	\mathds{P}\left( \kappa_N^{\omega} \ge \theta \nu L_N \right) \le \mathrm{e}^{- \nu L_N (1 - \theta+ \theta \ln \theta)}
	\end{align}
	if $\theta \ge 1$ and 
	\begin{align} \label{zxkjzxckhj 2}
	\mathds{P}\left( \kappa_N^{\omega} \le \theta \nu L_N \right) \le \mathrm{e}^{- \nu L_N (1 - \theta+ \theta\ln \theta)}
	\end{align}
	for $0 < \theta\le 1$, see \cite[Section 3.3.2]{SeiYngZag12}. The first part of this theorem now follows with the Borel--Cantelli lemma. 
	
	Furthermore, using \eqref{first part of Theorem k LN to nu} we conclude
 \begin{align*}
 \liminf\limits_{N \to \infty} \dfrac{\kappa_N^{\omega}}{L_N} \ge (1 - \epsilon)\nu \quad \text{ and } \quad \limsup\limits_{N \to \infty} \dfrac{\kappa_N^{\omega}}{L_N} \le (1 + \epsilon)\nu
 \end{align*}
 $\mathds P$-almost surely. Since $\epsilon > 0$ can be chosen arbitrarily, the last statement of this theorem is also shown.
\end{proof}

\begin{lemma} \label{Lemma hat lj equal lj 1}
For all $0 < \epsilon<1$ and for $\mathds P$-almost all $\omega \in \Omega$ there exists an $\widetilde N \in \mathds{N}$ such that for all $N \ge \widetilde N$,
$$\lNjomega \begin{cases}
 = \ljomega \quad & \text{ if } j \in J_{\lceil (1 - \epsilon) \nu L_N/2 \rceil} \\
 \le \ljomega\quad & \text{ if } j \in J_{\lfloor \nu L_N \rfloor} \\
 = 0 \quad & \text{ if } j \in \mathds Z \backslash ( J_{\lfloor \nu L_N \rfloor} \cup \{0\}) 
 \end{cases} \ , $$
 %where $\lNjomega = |(\hat x_j(\omega),\hat x_{j+1}(\omega)) \cap \Lambda_N|$ and $\ljomega := |(\hat x_j(\omega),\hat x_{j+1}(\omega))|$.
 %beachte: $M_j = 0$ für beide Randintervalle
%\begin{align*}
% \{ \ljomega \}_{j=-\lceil \nu L_N/4 \rceil}^{\lceil \nu L_N/4 \rceil} \subsetneq \{ \lNjomega \}_{j \in \mathds Z : M(\lNjomega) \ge 1} \subsetneq \{ \ljomega \}_{j=-\lfloor \nu L_N \rfloor}^{\lfloor \nu L_N \rfloor} \ .
 %\end{align*}
 where $J_k := \{-k, -k+1,\ldots,k-1,k\} \backslash\{0\}$, $k \in \mathds N$.
\end{lemma}
\begin{proof}
Recall that $I_N^{j,\omega} = (\hat x_j(\omega),\hat x_{j+1}(\omega)) \cap \Lambda_N$. Clearly, $\lNjomega \le \ljomega$ for all $j \in \mathds Z$, $N \in \mathds{N}$, and $\omega \in \Omega$. In addition, we have $\lNjomega = \ljomega$ for all $j \in \mathds Z$, $N \in \mathds{N}$, and $\omega \in \Omega$ for which $I_N^{j,\omega} \subset \Lambda_N$, that is, for which $I_N^{j,\omega}$ is entirely within the window $\Lambda_N$. Similarly, we have $\lNjomega = 0$ for every $j \in \mathds Z$, $N \in \mathds{N}$, and $\omega \in \Omega$ for which $I_N^{j,\omega} \cap \Lambda_N = \emptyset$, that is, for which $I_N^{j,\omega}$ is entirely outside the window $\Lambda_N$. 

We now determine when either of the last two cases holds. By $\kappa_N^{(1),\omega} \in \mathds{N}$, we denote the number of atoms of the Poisson random measure within $(-L_N/2,0]$. Similarly, the number of atoms of the Poisson random measure within $[0,L_N/2)$ shall be $\kappa_N^{(2),\omega} \in \mathds{N}$.
Let an arbitrary $0 < \epsilon < 1$ be given. Since there are not $\kappa_N^{\omega}$ but $\kappa_N^{\omega} + 1$ many intervals within $\Lambda_N$, we then pick an arbitrary $0 < \epsilon' < \epsilon$. With an appropriate version of Lemma~\ref{Theorem k LN to nu} we then conclude that for $\mathds P$-almost all $\omega \in \widetilde \Omega$ there exists an $\widetilde N \in \mathds{N}$ such that for all $N \ge \widetilde N$,
 $$\dfrac{1}{2}(1 - \epsilon')\nu L_N \le \kappa_N^{(1),\omega} \le \dfrac{1}{2}(1 + \epsilon') \nu L_N$$
 and
 $$\dfrac{1}{2}(1 - \epsilon')\nu L_N \le \kappa_N^{(2),\omega} \le \dfrac{1}{2}(1 + \epsilon') \nu L_N \ .$$ 
\end{proof}

\begin{lemma} \label{lemma llargest to infty almost surely}
	For all $0 < \epsilon < 1$ there exists a set $\widetilde \Omega \subset \Omega$ with $\mathds P(\widetilde \Omega) = 1$ and the following property: For every $\omega \in \widetilde \Omega$ there exists an $\widetilde N = \widetilde N(\epsilon,\omega)$ such that for all $N \ge \widetilde N$ one has
	\begin{align*}
	\llargest \ge \nu^{-1} \Big[ \ln( L_N) - (1 + \epsilon) \ln ( \ln(L_N)) \Big] \ .
	\end{align*}
\end{lemma}

\begin{proof}
 Let $0 < \epsilon < 1$ be arbitrarily given. As a first step, we conclude that
		 		\begin{align*}
	 		%& \mathds{P} \Bigg( \max \left\{ \lNjomega : j \in J_{\lceil (1-\epsilon) \nu L_N / 2 \rceil} \right\} \le \nu^{-1} \Big( \ln(L_N) - ( 1 + \epsilon) \ln ( \ln (L_N) ) \Big) \Bigg)\\
	 		& \mathds{P} \Bigg( \max \left\{ \ljomega : j \in J_{\lceil (1-\epsilon) \nu L_N / 2 \rceil} \right\} < \nu^{-1} \Big[ \ln(L_N) - ( 1 + \epsilon) \ln [ \ln (L_N) ] \Big] \Bigg)\\
		 		& \quad \le \left( 1 - \dfrac{[ \ln(L_N)]^{1 + \epsilon}}{L_N} \right)^{2 \lceil (1 - \epsilon) \nu L_N/ 2 \rceil} \ .
		 		\end{align*}
		 		Here, we used the fact that $\{ \ljomega : j \in \mathds Z \backslash \{0\} \}$ are mutually independent, exponentially distributed random variables.
		 		With the inequality $\ln(1 - x) \le -x$ for all $0 < x < 1$, we obtain the bound
		 		\begin{align*}
		 		2 \left\lceil \dfrac{(1 - \epsilon) \nu L_N}{2} \right\rceil \cdot \ln \left[ \left( 1 - \dfrac{[\ln(L_N)]^{1 + \epsilon}}{L_N} \right) \right] & \le - (1 - \epsilon) \nu [ \ln(L_N)]^{1 + \epsilon} \le - 2 \ln(N) 
		 		\end{align*}
		 		for all but finitely many $N \in \mathds{N}$. Therefore, we have
		 		\begin{align*}
		 		 \sum\limits_{N = 1}^{\infty} \left( 1 - \dfrac{[ \ln(L_N)]^{1 + \epsilon}}{L_N} \right)^{2 \lceil (1 - \epsilon) \nu L_N \rceil / 2} < \infty
		 		\end{align*}
		 		and, consequently, also
		 		\begin{align*}
		 		\sum\limits_{N = 1}^{\infty} \mathds{P} \Bigg( \max \left\{ \ljomega : j \in J_{\lceil (1-\epsilon) \nu L_N / 2 \rceil} \right\} < \nu^{-1} \Big( \ln(L_N) - ( 1 + \epsilon) \ln [ \ln (L_N) ] \Big) \Bigg) < \infty \ .
		 		\end{align*}
		 		By using Borel--Cantelli's lemma, %\textcolor{red}{Lemma ???}
		 		it now follows that there exists a set $\widetilde \Omega_1 \subset \Omega$ with $\mathds P(\widetilde \Omega_1) = 1$ and the following property. For every $\omega \in \widetilde \Omega_1$ there exists an $\widetilde N_1=\widetilde N_1(\epsilon,\omega) \in \mathds{N}$ such that
		 		\begin{align*}
		 		\max \left\{ \ljomega : j \in J_{\lceil (1-\epsilon) \nu L_N / 2 \rceil} \right\} \ge \nu^{-1} \Big\{ \ln( L_N) - (1 + \epsilon) \ln [ \ln(L_N)] \Big\}
		 		\end{align*}
		 		for all $N \ge \widetilde N_1$.
		 		
		 		On the other hand, according to Lemma~\ref{Lemma hat lj equal lj 1} there exists also a set $\widetilde \Omega_2 \subset \Omega$ with $\mathds P(\widetilde \Omega_2) = 1$ and the following property. For every $\omega \in \widetilde \Omega_2$ there exists an $\widetilde N_2 = \widetilde N_2(\epsilon, \omega) \in \mathds{N}$ such that for all $N \ge \widetilde N_2$ it is
		 		\begin{align*}
		 		\Big\{ \ljomega : j \in J_{\lceil (1-\epsilon) \nu L_N / 2 \rceil} \Big\} \subset \Big\{ \lNjomega : j \in \mathds Z \backslash \{0\} \Big\}\backslash \{0\} \ .
		 		\end{align*}
		 		and therefore
		 		\begin{align*}
		 		\max\Big\{ \ljomega : j \in J_{\lceil (1-\epsilon) \nu L_N / 2 \rceil} \Big\} \le \max\Big\{ \lNjomega : j \in \mathds Z \backslash \{0\} \Big\} = \llargest \ .
		 		\end{align*}

		 		Altogether, we thus have shown that for every $\omega \in \widetilde \Omega_1 \cap \widetilde \Omega_2$ there exists an $\widetilde N := \max\{ \widetilde N_1, \widetilde N_2\}$ such that for all $N \ge \widetilde N$ we have
		 				 		\begin{align*}
		 		\llargest \ge \max \left\{ \ljomega : j \in J_{\lceil (1-\epsilon) \nu L_N / 2 \rceil} \right\} \ge \nu^{-1} \Big\{ \ln( L_N) - (1 + \epsilon) \ln [ \ln(L_N)] \Big\} \ .
		 		\end{align*}
		 		Since $\mathds P(\widetilde \Omega_1 \cap \widetilde \Omega_2) = 1$, we have proved this theorem.
		 	\end{proof}

For the next lemma recall the definition of $\llargesttilde$ and $\llargesttildezwei$ in the proof of Theorem~\ref{Theorem gap condition LSmodel}.
		 	
\begin{lemma}\label{Theorem Abstand grosstes und viertgrosstes Intervall}
 For any $a > 0$ and $0 < \hat \eta < 1/2$, we have
 \begin{align} \label{Abstand grosstes und viertgrosstes Intervall}
\lim\limits_{N \to \infty} \mathds P \left( \llargesttilde - \llargesttildezwei > \dfrac{a}{N^{1-\hat \eta}}\right)= 1 \ .
\end{align}
\end{lemma} 
\begin{proof}
For all $k \in \mathds{N}$, let $\{ l^{j,\omega}\}_{j=1}^{k}$ be a set of $k$ independent and identically distributed random variables with common probability density $\nu \e^{-\nu l}$. For all $k \ge 2$, we define $\llargestk{1}$ and $\llargestk{2}$ as the largest and the second largest element of the set $\{ l^{j,\omega}\}_{j=1}^{k}$, respectively.
%We also define the set $A_N := \{ \ \llargesttilde, \llargesttildezwei \neq l_N^{(0),\omega} \}$. It is $\lim_{N \to \infty} \mathds P(A_N) = 1$; one can concluded this fact with Lemma~\ref{lower bound ground-state energy}, for example.
Then for any $2 \le k \in \mathds N$ and with
$$\widetilde \Omega_k := \left\{ \ \llargestk{1} - \llargestk{2} \ge \dfrac{a}{(k+1)^{1-\hat \eta}} \right\}$$
one has, see also \cite[Section 6.3]{LenobleZagrebnovLuttingerSy}, 
\begin{align*}
\mathds P(\widetilde \Omega_k) & = k(k-1) \int\limits_{a (k+1)^{-1+\hat \eta}}^{\infty} \int\limits_0^{x-a (k+1)^{-1+\hat \eta}} (1-e^{-\nu y})^{k-2} \, \nu \e^{-\nu x} \, \nu \e^{-\nu y}  \, \mathrm{d} y \, \mathrm{d} x\\
& = k \int\limits_{a (k+1)^{-1+\hat \eta}}^{\infty} \left( 1 - \e^{-\nu(x - a (k+1)^{-1 + \hat \eta})} \right)^{k-1} \nu \e^{-\nu x} \, \mathrm{d} x \\
%\shortintertext{with the substitution $u = 1 - \e^{-\nu(x - a(k+1)^{-1+\hat \eta})}$}
%= \e^{- \nu a (k+1)^{-1+\hat \eta}} k \int\limits_0^1 u^{k-1} \, \mathrm{d} u \\
& = \e^{- \nu a (k+1)^{-1+\hat \eta}} \ .
\end{align*}
We hence conclude that, see also \cite[Section 3.3.3]{SeiYngZag12}, %Beachte hier, dass $\kappa_N^{\omega} + 3 gewählt wurde, damit der Nenner hier niemals 0 sein kann!
\begin{align*}
& \lim\limits_{N \to \infty} \mathds P \left( \llargesttilde - \llargesttildezwei > \dfrac{a}{(\kappa_N^{\omega} + 1)^{1-\hat \eta}} \right) \\
%& \quad \ge \lim\limits_{N \to \infty} \sum\limits_{k=(1-L_N^{-\epsilon})\nu L_N}^{(1+L_N^{-\epsilon})\nu L_N} \mathds P \left( \left\{ \llargesttilde - \llargesttildezwei > \dfrac{a}{(\kappa_N^{\omega} + 1)^{1-\hat \eta}} \right\} \cap \left\{ \kappa_N^{\omega} = k \right\} \right) \\
& \quad \ge \lim\limits_{N \to \infty} \sum\limits_{k=(1-L_N^{-\epsilon})\nu L_N}^{(1+L_N^{-\epsilon})\nu L_N}\mathds P \left( \left\{ \llargestkminuszwei{1} - \llargestkminuszwei{2} > \dfrac{a}{(k + 1)^{1-\hat \eta}} \right\} \cap \left\{ \kappa_N^{\omega} = k \right\} \right) \\
& \quad \ge \lim\limits_{N \to \infty} \sum\limits_{k=(1-L_N^{-\epsilon})\nu L_N}^{(1+L_N^{-\epsilon})\nu L_N}\left[ \mathds P ( \widetilde \Omega_{k-2} ) + \mathds P \left( \kappa_N^{\omega} = k \right) - 1 \right] \\
& \quad = 1
\end{align*}
where $\hat \eta < \epsilon < 1/2$.
Here, we used the fact that
\begin{align*}
\lim\limits_{N \to \infty} \sum\limits_{k=(1-L_N^{-\epsilon})\nu L_N}^{(1+L_N^{-\epsilon})\nu L_N}\mathds P ( \kappa_N^{\omega} = k ) = 1 \ .
\end{align*}
since, see inequalities~\eqref{zxkjzxckhj 1} and~\eqref{zxkjzxckhj 2},
\begin{align*}
 \lim\limits_{N \to \infty} \mathds P \left( \kappa_N^{\omega} < (1 - L_N^{-\epsilon})\nu L_N \right) = \lim\limits_{N \to \infty} \mathds P \left( \kappa_N^{\omega} > (1 + L_N^{-\epsilon})\nu L_N \right) = 0 \ .
\end{align*}
In addition, 
\begin{align*}
\left| \lim\limits_{N \to \infty} \sum\limits_{k=(1-L_N^{-\epsilon})\nu L_N}^{(1+L_N^{-\epsilon})\nu L_N}\left[ \mathds P ( \widetilde \Omega_{k-2} ) - 1 \right] \right| & \le \lim\limits_{N \to \infty} \sum\limits_{k=(1-L_N^{-\epsilon})\nu L_N}^{(1+L_N^{-\epsilon})\nu L_N}\left| \mathds P ( \widetilde \Omega_{k-2} ) - 1 \right| \\
& \quad \le \lim\limits_{N \to \infty} \sum\limits_{k=(1-L_N^{-\epsilon})\nu L_N}^{(1+L_N^{-\epsilon})\nu L_N} \nu a(k-1)^{-1+\hat \eta} \\
& \quad \le \lim\limits_{N \to \infty} \nu a \dfrac{2 L_N^{-\epsilon} \nu L_N}{\big[ (1-L_N^{-\epsilon})\nu L_N - 1\big]^{1 -\hat \eta}} \\
& \quad \le \lim\limits_{N \to \infty} \nu a \dfrac{2 \nu L_N^{1 - \epsilon}}{\left( \dfrac{1}{2} \nu L_N \right)^{1 - \hat \eta}} \\
& \quad = 0 \ ,
\end{align*}
since $\hat \eta < \epsilon$. Lastly, again due to inequality~\eqref{zxkjzxckhj 2},

\begin{align*}
 & \lim\limits_{N \to \infty} \mathds P \left( \llargesttilde - \llargesttildezwei > \dfrac{a}{(2 \nu \rho^{-1} N)^{1-\hat \eta}} \right) \\
 %& \quad \ge \lim\limits_{N \to \infty} \mathds P \left( \left\{ \ \llargesttilde - \llargesttildezwei > \dfrac{a}{(2\nu \rho^{-1}N)^{1-\hat \eta}} \right\} \cap \Big\{ \ \kappa_N^{\omega} \le 2\nu \rho^{-1} N - 1 \Big\} \right) \\
 & \quad \ge \lim\limits_{N \to \infty} \mathds P \left( \left\{ \ \llargesttilde - \llargesttildezwei > \dfrac{a}{(\kappa_N^{\omega} + 1)^{1-\hat \eta}} \right\} \cap \Big\{ \ \kappa_N^{\omega} \le 2\nu \rho^{-1}N - 1 \Big\} \right) \\
& \quad \ge \lim\limits_{N \to \infty} \mathds P \left( \llargesttilde - \llargesttildezwei > \dfrac{a}{(\kappa_N^{\omega} + 1)^{1-\hat \eta}} \right) + \lim\limits_{N \to \infty} \mathds P \left( \kappa_N^{\omega} \le 2 \nu \rho^{-1}N - 1 \right) - 1\\
& \quad = 1 \ .
\end{align*}
\end{proof}
		 	
\newpage
{\small
	\bibliographystyle{amsalpha}
	\bibliography{mybibfile}}

\end{document}